\documentclass[pra,aps,superscriptaddress,twocolumn,nopacs,longbibliography,nofootinbib]{revtex4-1}	

\usepackage{amsthm,amsmath,mathtools,graphicx,subfigure,amsfonts,hyperref,xcolor,verbatim,hyperref,soul,ulem,bbold,enumitem,bbm,tikz,pgfplots}
\usepackage[latin1]{inputenc}

\hypersetup{
	colorlinks=true,   
	linkcolor=blue,    
	citecolor=blue,    
	urlcolor=blue      
}

\def\A{ {\cal A} }
\def\B{ {\cal B} }

\def\D{ {\cal D} }

\def\G{ {\cal G} }

\def\I{ {\cal I} }
\def\M{ {\cal M} }

\def\P{ {\cal P} }

\def\I{ {\cal I} }

\newcommand{\diag}[1]{\mathrm{diag}\left(#1\right)}
\def\>{\rangle}
\def\<{\langle}
\newcommand{\bra}[1]{\langle {#1} |}
\newcommand{\ket}[1]{| {#1} \rangle}
 
\newcommand{\ketbra}[2]{\ensuremath{\left|#1\right\rangle\!\!\left\langle#2\right|}}
\newcommand{\braket}[2]{\ensuremath{\left\langle#1\right|\!\!\left.#2\right\rangle}}
\newcommand{\matrixel}[3]{\ensuremath{\left\langle #1 \vphantom{#2#3} \right| #2 \left| #3 \vphantom{#1#2} \right\rangle}}
\newcommand{\tr}[1]{\mathrm{Tr}\left( #1 \right)}
\newcommand{\trr}[2]{\mathrm{Tr}_{#1}\left( #2 \right)}

\newcommand{\iden}{\mathbb{1}}
\newcommand{\vect}[1]{|#1\rangle\rangle}

\renewcommand{\v}[1]{\ensuremath{\boldsymbol #1}}

\DeclareMathOperator{\rank}{rank}
\newcommand{\1}{\mathbbm{1}}

\theoremstyle{plain}
\newtheorem{thm}{Theorem}
\newtheorem{lem}[thm]{Lemma}
\newtheorem{corol}[thm]{Corollary}
\newtheorem{prop}[thm]{Proposition}
\newtheorem{conj}[thm]{Conjecture}
\theoremstyle{definition}
\newtheorem{defn}[thm]{Definition}
\theoremstyle{remark}
\newtheorem{rmk}[thm]{Remark}

\begin{document}

\normalem

\title{Distinguishing classically indistinguishable states and channels}

\author{Kamil Korzekwa}
\affiliation{Centre for Engineered Quantum Systems, School of Physics, The University of Sydney, Sydney, NSW 2006, Australia}
\author{Stanis{\l}aw Czach{\'o}rski}
\affiliation{Faculty of Physics, Astronomy and Applied Computer Science, Jagiellonian University, 30-348 Krak{\'o}w, Poland}
\author{Zbigniew Pucha{\l}a}
\affiliation{Faculty of Physics, Astronomy and Applied Computer Science, Jagiellonian University, 30-348 Krak{\'o}w, Poland}
\affiliation{Institute of Theoretical and Applied Informatics, Polish Academy of Sciences, 44-100 Gliwice, Poland}
\author{Karol {\.Z}yczkowski}
\affiliation{Faculty of Physics, Astronomy and Applied Computer Science, Jagiellonian University, 30-348 Krak{\'o}w, Poland}
\affiliation{Center for Theoretical Physics, Polish Academy of Sciences, 02-668 Warszawa, Poland}

\begin{abstract}
	We investigate an original family of quantum distinguishability problems, where the goal is to perfectly distinguish between $M$ quantum states that become identical under a completely decohering map. Similarly, we study distinguishability of $M$ quantum channels that cannot be distinguished when one is restricted to decohered input and output states. The studied problems arise naturally in the presence of a superselection rule, allow one to quantify the amount of information that can be encoded in phase degrees of freedom (coherences), and are related to time-energy uncertainty relation. We present a collection of results on both necessary and sufficient conditions for the existence of $M$ perfectly distinguishable states (channels) that are classically indistinguishable. 
\end{abstract}

\date{\today}

\maketitle

\section{Introduction}

Finding optimal schemes for distinguishing between quantum states under various 
assumptions forms a family of important problems in quantum information 
science, with applications within quantum cryptography and quantum 
computation~\cite{Fuchs96,fuchs1999cryptographic,chefles2000quantum}. It is 
well known that two pure states can be deterministically discriminated if and 
only if they are orthogonal, or, in the case of mixed states, if their supports 
do not overlap~\cite{Hel76}. However, any interaction of the investigated 
system with an environment leads to the process of quantum decoherence, which 
reduces the probability of correctly distinguishing between given quantum states~\cite{nielsen2010quantum}. The full decoherence process can be described by a channel that sends any quantum state into a classical state represented by a corresponding diagonal density matrix. It may then happen that two orthogonal, completely distinguishable quantum states, decohere to the same classical state, e.g., qubit states \mbox{$|+\rangle\propto|0\rangle+|1\rangle$} and 
\mbox{$|-\rangle\propto|0\rangle-|1\rangle$} are orthogonal and decohere to the 
same maximally mixed classical state. One can therefore study the deteriorating 
effect of the decoherence process on quantum information by asking: how many 
perfectly distinguishable quantum states decohere to a fixed classical state?

More formally, in the first part of this paper we introduce and investigate the problem of distinguishing quantum states $\{\rho^{(n)}\}_{n=1}^M$ that are classically indistinguishable, i.e., their decohered versions, $\{\D(\rho^{(n)})\}_{n=1}^M$ with $\D$ denoting a \emph{completely decohering} quantum channel in the preferred basis $\{\ket{k}\}_{k=1}^d$, cannot be distinguished with a probability larger than $1/M$ (which corresponds to a random guess). Such states share the same \emph{classical version}~$\v{p}$,
\begin{equation}
	\forall n:\quad \bra{k}{\rho^{(n)}}\ket{k}=:p_k,
\end{equation}
and the main object of our studies is thus defined as follows.

\begin{defn}[\mbox{$M$-distinguishability} region] 
	A $d$-dimensional probability vector $\v{p}$ belongs to \mbox{$M$-distinguishability} region $\A_d^M$ of the probability simplex $\Delta_d$ if and only if there exist $M$ perfectly distinguishable quantum states with the same classical version~$\v{p}$.
\end{defn}

Our interest in the mathematical structure of \mbox{$M$-distinguishability} regions is physically motivated by its direct relation to the problem of encoding information in coherence. Note that $M$ perfectly distinguishable states allow one to encode $\log_2 M$ bits of information. By fixing the classical degrees of freedom (the classical version $\v{p}$) for a set of states $\{\rho^{(n)}\}_{n=1}^M$, the only way left to encode information is to use the quantum degrees of freedom (coherence). Thus, the maximal number $M$ of perfectly distinguishable states with a fixed classical version~$\v{p}$ quantifies the capacity of coherence to carry information that cannot be accessed classically. This is similar in spirit to the problem of quantum data hiding~\cite{terhal2001hiding,divincenzo2002quantum}, when one wants to store classical bits in correlations, so that they are inaccessible locally. Also, the separation into the classical and quantum degrees of freedom for encoding information is reminiscent of the previous studies on splitting uncertainty into classical and quantum part parts~\cite{luo2005quantum,korzekwa2014quantum}.

It is important to note that the restriction to classical version of a state is not only an abstract constraint allowing one to assess the ability of coherence to carry information. Whenever the dynamics obeys a symmetry linked to some conservation law, the processing of states that break this symmetry is constrained~\cite{marvianthesis,marvian2014extending}. In particular, since coherence in the energy eigenbasis breaks time-translation symmetry, the conservation of energy restricts possible processing of coherences~\cite{lostaglio2015description,lostaglio2015quantum}. As a result, without the access to an additional resource in the form of a quantum reference frame for phase~\cite{bartlett2007reference}, states $\rho$ and its decohered version $\D(\rho)$ become indistinguishable\footnote{More precisely, here the decohering channel $\D$ destroys coherence between different energy eigenspaces, leaving the coherence between Hamiltonian eigenstates corresponding to the same eigenvalue unchanged. Therefore, our studies apply to this scenario when the Hamiltonian of the system is non-degenerate.}, and so one can access only information encoded in the classical degrees of freedom (with the distinguished basis given by the energy eigenbasis). Let us point out that this indistinguishability plays a crucial role within quantum thermodynamics as it affects the amount of work that can be extracted from a system prepared in a superposition of energy eigenstates~\cite{korzekwa2016extraction}. 

One can also invert the question and instead of asking how much information can 
be encoded in coherence, ask: how much information is lost due to the 
irreversible process of decoherence? One way to quantify the deteriorating 
effect of the decohering channel $\D$ is to ask about the largest number $M$ of 
messages that could have been encoded in $\D(\rho)$ before the action of~$\D$. 
In other words, one is interested in finding the number of orthogonal preimages 
of $\D(\rho)$, known as \emph{coherifications} of 
$\D(\rho)$ (see Ref.~\cite{korzekwa2018coherifying} for details and Appendix~\ref{app:coher} for an intuitive visualization of the coherification procedure). It is also worth noting that since 
$\D$ describes the process of measuring the system in a given basis and then 
discarding the result, $M$-distinguishability regions can shed new light on the 
disturbing effect measurements have on a quantum system.

Finally, there is a strong link between the problem of $M$-distinguishability 
and energy-time uncertainty relation. For this, consider now that the 
distinguished basis is given by the eigenstates of Hamiltonian $H$, 
$\{\ket{E_k}\}$, so that $\v{p}$ is given by 
\mbox{$p_k=\matrixel{E_k}{\rho}{E_k}$}. Although an observable for time does 
not exist, there is nevertheless the expectation that time and energy should be 
complementary variables, resulting in a version of uncertainty relation between 
them. Non-rigorously, it should state that if a given state $\rho$ has a 
well-defined energy then it is a bad clock, i.e., it does not significantly 
change in time (in the limit of $\rho$ being a sharp energy eigenstate, $\rho$ 
becomes stationary and does not evolve in time at all); and if a state $\rho$ 
allows one to distinguish different moments of time with high precision, then 
the energy of $\rho$ cannot be well-defined. Of course, there are many ways to 
quantify both the sharpness of energy of $\rho$ and the quality of $\rho$ as a 
clock. For example, in the most traditional formulation by Mandelstam and 
Tamm~\cite{mandelstam1991uncertainty}, the uncertainty of energy is quantified 
by the variance of~$\v{p}$, and the timing quality of $\rho$ is given by the 
minimal time needed for $\rho$ to evolve to another distinguishable state 
(clearly, if such time is long, then the time resolution is low, meaning the 
quality of $\rho$ as a clock is low). The maximal number $M$ of perfectly 
distinguishable states with a fixed diagonal $\v{p}$ can now be related to a 
particular version of the energy-time uncertainty relation presented above. 
Namely, given a state with energy distribution $\v{p}$, its timing quality can 
be measured by $M$, which tells us how many different moments in time can be 
distinguished unambiguously (i.e, with no uncertainty) using $\rho$. The 
$M$-distinguishability regions $\A_d^M$ provide then a geometric way to 
visualize energy-time uncertainty relation: the closer one gets to the centre 
of the probability simplex (the uniform distribution), the more uncertain the 
energy outcomes become, but the better potential timing quality of the state 
becomes.

In the second part of the paper we focus on a closely related notion of classically indistinguishable channels, by studying the distinguishability of their coherified versions~\cite{korzekwa2018coherifying}. Research along
this line was recently performed for the problem of discriminating quantum measurements~\cite{puchala2018strategies,puchala2018multiple}, where it was shown that the diamond norm distance between two von Neumann measurements is given by the minimal value of the distance between their completely coherified versions. Here, we consider classically indistinguishable channels, which are the channels that cannot be distinguished by using classical input states and being restricted to the classical versions of output states. A set of quantum channels $\{\Phi^{(n)}\}_{n=1}^M$ that are classically indistinguishable share the same \emph{classical action}, so they generate the same stochastic matrix $T$,  
\begin{equation}
	\forall n:\quad \bra{k}{\Phi^{(n)}(\ketbra{l}{l})}\ket{k}=:T_{kl},
\end{equation}
which describes discrete dynamics in the probability simplex. By allowing access to arbitrary input states (including entangled ones) and general quantum measurements of the output states, such channels can potentially be distinguished. The natural question that arises then, and that we address in the paper, is: how many perfectly distinguishable quantum channels can there be that share the same classical action~$T$? More formally, we study \emph{distinguishability numbers} defined in the following way.
\begin{defn}[Distinguishability numbers]
	Distinguishability number $\M(T)$ is the maximal number of quantum channels that share the same classical action~$T$ and can be perfectly distinguished. Restricted distinguishability number $\tilde{\M}(T)$ is the maximal number of quantum channels that share the same classical action~$T$ and can be perfectly distinguished without using entangled input states.
\end{defn}

Studying distinguishability numbers allows one to quantify distinct ways of processing information encoded in coherences. More precisely, classically indistinguishable channels transform classical degrees of freedom in the same way, described by the fixed classical action $T$, and so the only way to distinguish them is through the effect they have on quantum degrees of freedom, i.e., coherences. As with the quantum states, here also we can draw an analogy with the entanglement scenario in which one wants to investigate quantum channels that cannot be distinguished by scrutinizing local systems~\cite{fan2004distinguishability}, as they transform local states in the same way. Instead of the locality constraint, here we focus on classicality constraint that can arise, e.g., due to the conservation law and a lack of an appropriate reference frame~\cite{bartlett2007reference}. In such situations one can only prepare input classical states and cannot distinguish between output states that share the same classical version. Therefore, effectively one only has access to the classical action $T$ and cannot distinguish channels corresponding to the same stochastic matrix~$T$.

One can also use distinguishability numbers to get insight into the effect that intermediate measurements have on discrete quantum Markov processes. Imagine the scenario in which the system undergoes a discrete process that at each time step transforms it according to a fixed quantum channel $\Phi$. Moreover, assume that before and after each application of $\Phi$ one observes the system by measuring it in the preferred basis $\{\ket{k}\}$. This way, by repeating the experiments many times and recording measurements outcomes, one can reconstruct the transition matrix $T$ between different states~$\ket{k}$. Now, despite the fact that there may be a whole family $\{\Phi^{(n)}\}$ of quantum processes leading to the same observations, sequential measurements collapse all $\Phi^{(n)}$ to the same classical Markov process described by $T$. If one did not observe the system at each time step, the accumulated interference effects could result in each $\Phi^{(n)}$ transforming the system in completely distinct way, so that by properly measuring the final state one could find out which $\Phi^{(n)}$ actually happened. The distinguishability number $\M(T)$ describes then the number of quantum Markov processes that are equal and equivalent to a classical process $T$ if observed at each time step, but completely distinct if unobserved.

The paper is structured in the following way. First, in Sec.~\ref{sec:setting}, we set the scene by introducing necessary concepts, fixing the notation and formally defining the notion of state and channel distinguishability. Then, Section~\ref{sec:states} is devoted to the studies of distinguishability of classically indistinguishable states, while Section~\ref{sec:channels} focuses on classically indistinguishable channels. Finally, the conclusions and open problems for future research can be found in Sec.~\ref{sec:outlook}.

\section{Setting the scene}
\label{sec:setting}

\subsection{Mathematical background and notation}

A state of a finite-dimensional quantum system is described by a density operator $\rho$ acting on a $d$-dimensional Hilbert space ${\cal H}_d$ that is positive, $\rho\geq 0$, and normalized by a trace condition, $\tr{\rho}=1$. A state is pure if $\rho=\rho^2$, so it can be represented by a \mbox{$1$-dimensional} projector, \mbox{$\rho=\ketbra{\psi}{\psi}$}; and mixed otherwise. General evolution of quantum states can be described by quantum channels, i.e., completely positive trace preserving (CPTP) maps acting on density matrices of order $d$. Every quantum channel $\Phi$ admits a Kraus decomposition~\cite{nielsen2010quantum} of the form
\begin{equation}
	\label{eq:kraus_decomposition}
	\Phi(\cdot)=\sum_k K_k(\cdot)K_k^\dagger,
\end{equation}
where $K_k$ are called Kraus operators and, due to trace preserving condition, satisfy \mbox{$\sum_kK_k^\dagger K_k=\iden$} with $\iden$ denoting the identity matrix of size $d$. Moreover, with each channel $\Phi$ one can associate a \emph{Jamio{\l}kowski state}~\cite{jamiolkowski1972linear}, defined by the image of the extended map acting on a maximally entangled state,
\begin{equation}
	\label{eq:jamiol}
	J_{\Phi} = \frac{1}{d}(\Phi \otimes {\I})\ketbra{\Omega}{\Omega},
\end{equation}
with \mbox{$\ket{\Omega}=\sum_k\ket{kk}$} and $\I$ denoting the identity channel. Under this isomorphism the CP condition is translated into positivity of $J_\Phi$, and the TP condition is replaced by $\trr{1}{J_\Phi}=\iden/d$.

The subset of \emph{classical states} is given by quantum states $\rho$ that are incoherent with respect to a given distinguished orthonormal basis $\{\ket{k}\}_{k=1}^d$, i.e., $\matrixel{k}{\rho}{l}=0$ for $k\neq l$. The choice of the basis is physically motivated by a particular problem under study, e.g., within quantum thermodynamics one is concerned with energy eigenbasis~\cite{lostaglio2015description,lostaglio2015quantum}. Classical state can be alternatively represented by a probability distribution $\v{p}=\diag{\rho}$, where $\diag{\rho}$ denotes a mapping of a density matrix $\rho$ into a probability vector $\v{p}$ with $p_k=\rho_{kk}$. Moreover, for a general quantum state $\rho$ we call the probability distribution $\diag{\rho}$ the \emph{classical version} of $\rho$. Note that under the \emph{completely decohering} quantum channel~$\D$, 
\begin{equation}
	\label{eq:decoh_state}
	\D(\rho)=\sum_k \matrixel{k}{\rho}{k}\ketbra{k}{k},
\end{equation}
every quantum state $\rho$ is mapped to a classical state specified by the classical version of $\rho$.

We also define a subset of \emph{classical channels} that consists of all channels $\Phi$ whose corresponding Jamio{\l}kowski states are classical, i.e., $\matrixel{kk'}{J_\Phi}{ll'}=0$ whenever $k\neq l$ or $k'\neq l'$. Classical channel can be alternatively represented by a stochastic transition matrix $T$ given by $\frac{1}{d}|T\rangle\rangle=\diag{J_\Phi}$, where $|\cdot\rangle\rangle$ denotes the (row-wise) \emph{vectorization} of a matrix,
\begin{equation}
	\label{eq:vectorization}
	|T\rangle\rangle:=(T\otimes\iden)\ket{\Omega}=\left( \1 \otimes T^\top \right) \ket{\Omega},
\end{equation}
and $T$ satisfies \mbox{$T_{kl}\geq 0$} and \mbox{$\sum_k T_{kl}=1$}. Moreover, for a general quantum channel $\Phi$ we call the corresponding transition matrix $T$ the \emph{classical action} of $\Phi$. A quantum channel $\Phi$ can be mapped to its classical version via a \emph{completely decohering supermap} that decoheres the corresponding Jamio{\l}kowski state $J_\Phi$~\cite{korzekwa2018coherifying}, and is described by the following two-step concatenation
\begin{equation}
	\Phi \rightarrow \Phi^\D= \D\circ\Phi\circ\D.
\end{equation}
Note also that the classical action $T$ of a channel $\Phi$ describes the transition between diagonal states, 
\begin{equation}
	\label{eq:transition}
	T_{kl}=\bra{k}\Phi(\ketbra{l}{l})\ket{k}.
\end{equation}
Therefore, a classical channel represented by $T$ maps a quantum state with classical version $\v{p}$ to a classical state $T\v{p}$; and a quantum channel $\Phi$ with classical action $T$ maps a classical state $\v{p}$ to a quantum state with classical version given by $T\v{p}$. Finally, a stochastic matrix $T$ is called bistochastic if $\sum_l T_{kl}=1$; and unistochastic if there exists a unitary matrix $U$ such that $T=U\circ \bar{U}$, with $\circ$ representing the entry-wise product (also known as Hadamard or Schur product).

Throughout the paper the dimension of the underlying Hilbert space will be denoted by $d$, so all operators (matrices) will act on $d$-dimensional state vectors, while quantum channels will act on $d\times d$ density matrices. The \mbox{$(d-1)$-dimensional} probability simplex that represents the set of $d$-dimensional classical states will be denoted by $\Delta_d$, while its centre, i.e., the maximally mixed distribution with each entry equal to $1/d$, will be denoted by~$\v{\eta}$. Moreover, we introduce a flat probability vector $\v{v}^M$ with first $M$ entries equal to $1/M$ (in particular \mbox{$\v{v}^d=\v{\eta}$}). Beyond the identity matrix and identity channel, $\iden$ and $\I$, we will make frequent use of the unitary Fourier matrix $F$ and the maximally mixing van der Waerden matrix $W$ defined by
\begin{equation}
	F_{kl}=\frac{1}{\sqrt{d}} \exp\left(\frac{2\pi i(k-1)(l-1)}{d}\right),\quad 	W_{kl}=\frac{1}{d},
\end{equation}
so that $|F_{kl}|^2=W_{kl}$. We also define a set of $d$ diagonal unitary matrices $D^{(k)}$, with the diagonal specified by the columns of $F$, i.e.,
\begin{equation}
	\label{eq:diagonal}
	D^{(k)}_{ll}=\sqrt{d}F_{kl}.
\end{equation}

\subsection{Distinguishability problem}

The central problem studied in this work concerns state and channel 
distinguishability, which are defined as follows. Given a quantum state $\rho$ 
and a promise that it belongs to a preselected set of $M$ states 
$\{\rho^{(n)}\}$ (with each one being equally likely), the task is to find the 
optimal way of deciding $n^*$ satisfying $\rho=\rho^{(n^*)}$. The optimality of 
the protocol means succeeding with the highest possible probability (and thus 
the problem is often referred to as the \emph{maximum likelihood} 
distinguishability). A similar question can be posed for quantum channels: 
given a single use of a channel $\Phi$, decide which one from the predefined 
set of $M$ equally likely channels $\{\Phi^{(n)}\}$ it is. We say that a set of 
$M$ states (channels) is \emph{$M$-distinguishable} if it admits perfect 
distinguishability, i.e., if there exists a protocol that succeeds with unit 
probability.

Let us first briefly discuss the simplest case of distinguishability problem for $M=2$. Given two classical states represented by probability distributions $\v{p}$~and~$\v{q}$, one finds that the maximum likelihood probability $P(\v{p},\v{q})$ of the correct distinction between them is given by
\begin{equation}
	\label{eq:classical_state_distinguish}
	P(\v{p},\v{q})=\frac{1+\delta(\v{p},\v{q})}{2},\quad \delta(\v{p},\v{q}):=\frac{1}{2}\sum_k|p_k-q_k|,
\end{equation}
with $\delta$ known as the \emph{total variation distance}. The optimal protocol simply consists of measuring the system in the distinguished basis, and upon observing outcome $k$ answer $\v{p}$ if $p_k\geq q_k$, and $\v{q}$ otherwise. Similarly, given two quantum states, ${\rho}$ and ${\sigma}$, the optimal measurement leads to probability $P(\rho,\sigma)$ of distinguishing them given by~\cite{nielsen2010quantum}
\begin{equation}
	\label{eq:quantum_state_distinguish}
	\!\! P(\rho,\sigma)=\frac{1+D_{\mathrm{tr}}(\rho,\sigma)}{2},\quad\!\! D_{\mathrm{tr}}(\rho,\sigma)=\frac{1}{2}\tr{|\rho-\sigma|},
\end{equation}
with $D_{\mathrm{tr}}$ known as the \emph{trace distance}. We conclude that two 
states are perfectly distinguishable if and only if they have orthogonal 
supports. This fact straightforwardly leads to the following statement: a set 
of $M$ states is \mbox{$M$-distinguishable} if and only if each pair of states have orthogonal supports.

Let us now proceed to channel distinguishability. To distinguish two classical 
channels represented by transition matrices $T^{(1)}$ and $T^{(2)}$, one has to 
find a classical state $\v{p}$ that optimizes the distinguishability between 
$T^{(1)}\v{p}$ and $T^{(2)}\v{p}$. Using convexity one can argue that such an 
optimal classical state should be sharp, i.e., it has all zero entries except 
for some $k$, for which it is equal to 1. Such a state is then transformed by 
$T^{(1)}$ to the $k$-th column of $T^{(1)}$, denoted by $T^{(1)}_{\star k}$; 
and by $T^{(2)}$ to the $k$-th column of $T^{(2)}$, denoted by $T^{(2)}_{\star 
k}$. Hence, the optimal probability of distinguishing $T^{(1)}$ from $T^{(2)}$ is given by
\begin{equation}
	\label{eq:classical_channel_distinguish}
	P(T^{(1)},T^{(2)})=\max_k\frac{1+\delta(T^{(1)}_{\star k},T^{(2)}_{\star k})}{2}.
\end{equation}
The problem of distinguishing between general quantum channels $\Phi^{(1)}$ and $\Phi^{(2)}$ is more complicated, due to the possible use of entangled states. Let us thus first consider that one has no access to entanglement. Then, analogously to the classical case, one has to find a quantum state $\rho$ that optimizes the distinguishability between $\Phi^{(1)}(\rho)$ and $\Phi^{(2)}(\rho)$. Again, using convexity argument, one can restrict the optimization to pure states $\psi$ leading to
\begin{equation}
	\tilde{P}(\Phi^{(1)},\Phi^{(2)})=\max_{\psi}\frac{1+D_{\mathrm{tr}}(\Phi^{(1)}(\psi),\Phi^{(2)}(\psi))}{2},
\end{equation}
where tilde denotes the constrained optimization with no entanglement. More fundamentally, however, one can make use of entangled states to improve the distinguishability between $\Phi^{(1)}$ and $\Phi^{(2)}$, so that
\begin{align}
	&\!\!\!\!\!\!\! P(\Phi^{(1)},\Phi^{(2)})\nonumber\\
	&\!\!\!=\max_{\Psi}\frac{1+D_{\mathrm{tr}}[(\Phi^{(1)}\otimes\I)(\Psi),(\Phi^{(2)}\otimes\I)(\Psi)]}{2},\!\!
\end{align}
where the optimization is over pure bipartite states $\Psi$.

\section{Classically indistinguishable states}
\label{sec:states}

\subsection{Permutohedron bound and $M$-distinguishability}

We start our analysis by finding necessary conditions for $M$-distinguishability. Geometrically this problem is equivalent to bounding $M$-distinguishability regions $\A^M_d$ within the probability simplex $\Delta_d$. First, note that, by definition, we have $\A_{d}^{l} \subset \A_{d}^{k}$ for \mbox{$1\leq k< l\leq d$}, and \mbox{$\A_{d}^{1} = \Delta_{d}$}. Now, in order to find further non-trivial conditions we introduce the concept of permutohedron~\cite{postnikov2009permutohedra}:
\begin{defn}[Permutohedron]
	For every $\v{x}\in\mathbb{R}^d$ the permutohedron $\P_d (\v{x})$ is the convex hull of all the permutations of $\v{x}$,
	\begin{align}
	\v{y}\in \P_d(\v{x}) \quad\Longleftrightarrow\quad \exists \v{\lambda}:~\v{y}=\sum_k \lambda_k \Pi_k \v{x},
	\end{align}
	with $\v{\lambda}$ being $d!$-dimensional probability vector and $\{\Pi_k\}$ denoting the set of $d!$ permutation matrices acting on \mbox{$d$-dimensional} vectors. In particular, we will use a shorthand notation $\P_d^M$ while referring to $\P_d(\v{v}^M)$.
\end{defn}

Permutohedra $\P_d^M$ form nested convex polytopes in $\mathbb{R}^{d-1}$ with vertices given by $\{\Pi_k \v{v}^M\}$ and satisfying \mbox{$\P_d^{M+1}\subset \P_d^{M}$, $\P_d^1=\Delta_d$} and $\P_d^d=\{\v{\eta}\}$. We illustrate the first few of them in Fig.~\ref{fig:permutohedra}. Using $\P_d^M$ we can now bound $\A_d^M$ via the following result.
\begin{prop}[Permutohedron bound]
	\label{prop:necessary}
	The necessary condition for $M$-distinguishability of a set of quantum states $\{\rho^{(n)}\}^M_{n=1}$ with a fixed classical version $\v{p}$ is $\max_k p_k \leq 1/M$. Equivalently, $\A_d^M \subseteq \P_d^M$.	
\end{prop}
\begin{proof} 
	Since $\{\rho^{(n)}\}^M_{n=1}$ are orthogonal, we have
	\begin{equation}
		\label{eq:dist_condition}
		\sum_{n=1}^M \rho^{(n)} \leq \iden.
	\end{equation}
	By taking the matrix element $\bra{k}\cdot\ket{k}$ of both sides we get
	\begin{equation}
		M p_k \leq 1,
	\end{equation}	
	which holds for all $k$, so in particular $\max_k p_k \leq 1/M$.
\end{proof}

\begin{figure}[t]
	\includegraphics[width=\columnwidth]{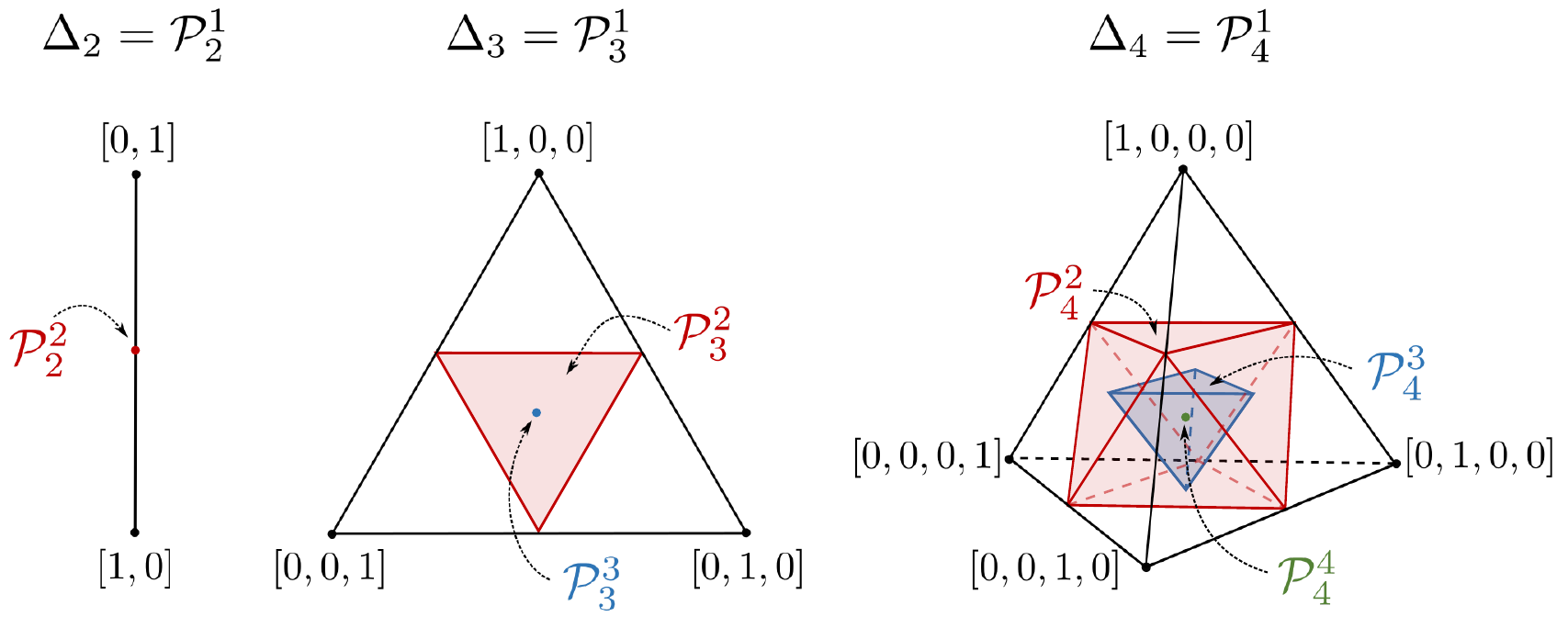}
	\caption{\label{fig:permutohedra} \emph{Permutohedra.} Visualization of permutohedra $\P_d^M$ for $2\leq d\leq 4$ and $1\leq M\leq d$. Permutohedron $\P_d^M$ has $\binom{d}{M}$ vertices located at the centres of $(M-1)$-faces of the probability simplex $\Delta_d$}
\end{figure}

Direct application of the above result to time-energy uncertainty scenario, 
described in the Introduction, leads to the following statement: a state that is able to distinguish $M$ different moments in time satisfies the inequality for min-entropy $H_\infty(\v{p})\geq \log M$, with $\v{p}$ denoting its distribution over energy. We note that this coincides with the particular version of the recent result presented in 
Ref.~\cite{coles2018entropic}, where the authors studied entropic formulations 
of energy-time uncertainty relation. Thus, any improvements over the 
permutohedron bound could also tighten inequalities derived there.

Before proceeding let us also state two useful results concerning 
$M$-distinguishability of pure states. First, we can relate it to 
the existence problem of particular unistochastic matrix.
\begin{lem}
	\label{lem:dist_uni}
	$M$-distinguishability of a set of pure quantum states $\{\ket{\psi}^{(n)}\}^M_{n=1}$ with a fixed classical version $\v{p}$ is equivalent to the existence of a unistochastic matrix $T$ with first $M$ columns equal to $\v{p}$.
\end{lem}
\begin{proof}
	First, assume that there exists a set of $M$ distinguishable pure states with a fixed classical action $\v{p}$, i.e., there exists $\{\ket{\psi^{(n)}}\}_{n=1}^M$ satisfying
	\begin{equation}
		\forall m,n:\quad |\langle k\ket{\psi^{(n)}}|^2=p_k,\quad \langle \psi^{(m)}\ket{\psi^{(n)}}=\delta_{mn},
	\end{equation}
	with $\delta_{mn}$ denoting Kronecker delta. Now, since $\{\ket{\psi^{(n)}}\}$ form an orthonormal set, one can construct a unitary $U$ with the first $M$ columns given by the components of these states. More precisely, we can define $U$ by
	\begin{equation}
		U_{kn}=\langle k\ket{\psi^{(n)}}
	\end{equation}
	for $k\in\{1,\dots, d\}$, $n\in\{1,\dots,M\}$, and complete the remaining columns with orthonormal states. Then the stochastic matrix $T=U\circ \bar{U}$ is unistochastic by definition, and $T_{kn}=p_k$ for $n\in\{1,\dots,M\}$.
	
	Conversely, assume that there exists a unistochastic $T$ with $T_{kn}=p_k$ for $n\in\{1,\dots,M\}$. This is equivalent to the existence of a unitary $U$ with the first $M$ columns given by
	\begin{equation}
		\ket{u_n}=\sum_k \sqrt{p_k} e^{i\phi_{kn}} \ket{k}.
	\end{equation}
	Since the columns of a unitary matrix are orthogonal the set $\{\ket{u_n}\}_{n=1}^M$ forms $M$ perfectly distinguishable quantum states with a fixed diagonal $\v{p}$.
\end{proof}

Moreover, we can show that for distributions lying at the boundary of permutohedron $\P_d^M$, $M$-distinguishability of mixed states is equivalent to $M$-distinguishability of pure states.
\begin{lem}
	\label{lem:boundary}
	Consider $\v{p}$ such that $p_{k^*}=1/M$ for some~$k^*$, i.e., $\v{p}$ lies at the boundary of a permutohedron $\P_d^M$. Then, $\v{p}\in\A_d^M$ implies the existence of $M$ orthogonal pure states with a fixed classical version $\v{p}$.
\end{lem}
\begin{proof}
	Assumption $\v{p}\in\A_d^M$ means that there exists a set $\{\rho^{(n)}\}_{n=1}^M$ of perfectly distinguishable quantum states. Let us diagonalize each $\rho^{(n)}$,
	\begin{equation}
		\rho^{(n)}=\sum_{\alpha=1}^{r_n} \lambda^{(n)}_\alpha \ket{\psi^{(n)}_\alpha}\bra{\psi^{(n)}_\alpha}.
	\end{equation}
	Now, on the one hand we get
	\begin{equation}
		\label{eq:lem_cond1}
		\forall n:~\frac{1}{M}=\bra{k^*}\rho^{(n)}\ket{k^*}=\sum_{\alpha=1}^{r_n}\lambda^{(n)}_\alpha |\bra{k^*}\psi^{(n)}_\alpha\rangle|^2,  
	\end{equation}
	resulting in
	\begin{equation}
		\sum_{n=1}^M\sum_{\alpha=1}^{r_n}\lambda^{(n)}_\alpha |\bra{k^*}\psi^{(n)}_\alpha\rangle|^2=1.
	\end{equation}	
	On the other hand, perfect distinguishability of $\{\rho^{(n)}\}_{n=1}^M$ implies 
	\begin{equation}
		\bra{\psi^{(m)}_\beta}\psi^{(n)}_\alpha\rangle = \delta_{mn} \delta_{\alpha\beta},
	\end{equation}
	so that one $\{\ket{\psi^{(n)}_\alpha}\}$ can be used to form a unitary matrix $U$ as in the proof of Lemma~\ref{lem:dist_uni}. Using analogous argument of the unistochasticity of $T=U\circ \bar{U}$ we then get
	\begin{equation}
		\label{eq:lem_cond2}	
		\sum_{n=1}^M \sum_{\alpha=1}^{r_n} |\bra{k^*} \psi^{(n)}_\alpha \rangle|^2\leq 1
	\end{equation}
	Finally, comparing Eqs.~\eqref{eq:lem_cond1}~and~\eqref{eq:lem_cond2}, we see that for all $n$ the spectrum $\lambda^{(n)}_\alpha$ is sharp, i.e., each $\rho^{(n)}$ is a pure state.
\end{proof}

\subsection{Tightness of the permutohedron bound}

We now proceed to analysing how tight the permutohedron bound is. We start with the following tightness result.
\begin{prop}
	\label{prop:sufficient}
	 The necessary condition for \mbox{$M$-distinguishability}, as stated by Proposition~\ref{prop:necessary}, is also sufficient for $M=2$ and $M=d$. Equivalently, $\A^2_d = \P^2_d$ and \mbox{$\A^d_d = \P^d_d$}.
\end{prop}

\begin{proof}
	We first show $\A^2_d = \P^2_d$. We need to prove that for a given $\v{p}$ the condition \mbox{$\max_k p_k\leq 1/2$} implies the existence of two perfectly distinguishable quantum states with a fixed classical version $\v{p}$. Consider the following two pure states,
	\begin{equation}
		\ket{\psi^{(1)}} = \sum_{k=1}^d \sqrt{p_k} \ket{k},\quad\ket{\psi^{(2)}} = \sum_{k=1}^d \sqrt{p_k} e^{i \phi_k} \ket{k},
	\end{equation}
	so that their overlap is given by	
	\begin{equation}
		|\!\bra{\psi^{(1)}}\psi^{(2)}\rangle\!|^2=\sum_{k=1}^d p_k e^{i\phi_k}.
	\end{equation}
	Now, note that the existence of phases $\{\phi_k\}$ such that the above expression vanishes is equivalent to the possibility of constructing a closed polygon out of $d$ segments of lengths $\{p_k\}$. Recall that the generalized triangle inequality states that the longest side of the polygon has to be shorter than the sum of the remaining sides; and its converse ensures that one can build a closed polygon if this condition is satisfied. Therefore, if
	\begin{equation}
		\max_k p_k\leq \sum_{k} p_k -\max_k p_k= 1-\max_k p_k,
	\end{equation}
	meaning $\max_k p_k\leq \frac{1}{2}$, then there exists a choice of phases $\{\phi_k\}$ such that the overlap between $\ket{\psi^{(1)}}$ and $\ket{\psi^{(2)}}$ vanishes.
	 
	 We now show that $\A^d_d = \P^d_d$. For this, we need to prove the existence of $d$ orthogonal states with classical version $\v{\eta}$. This is simply accomplished by choosing columns of the Fourier matrix, \mbox{$\ket{\psi^{(k)}}=F\ket{k}$}, which are all mutually orthogonal and the corresponding classical states are maximally mixed.
\end{proof}

Although, for the particular cases of $M=2$ and $M=d$, $M$-distinguishability regions $\A^M_d$ coincide with the corresponding bounding permutohedra $\P^M_d$, the following result shows that, in general, the permutohedron bound is not tight.
\begin{prop}
	\label{prop:insufficient}
	The necessary condition for \mbox{$M$-distinguishability}, as stated by Proposition~\ref{prop:necessary}, is not sufficient for $M=d-1$ and even $d>2$. Equivalently, $\A^{d-1}_d\neq \P^{d-1}_d$ for even $d>2$.
\end{prop}
The proof of the above result can be found in Appendix~\ref{app:insufficient}. We thus see, that the regions $\A_d^M$ have a more complex structure than permutohedra $\P_d^M$. Let us illustrate this using the simplest non-trivial example of $d=4$ and $M=3$. As shown in Fig.~\ref{fig:permutohedra}, the probability simplex $\Delta_4$ can be represented by a 3-dimensional tetrahedron, with maximally mixed distribution $\v{\eta}$ in the centre and vertices corresponding to sharp probability distributions, i.e., $(1,0,0,0)$ and permutations thereof. Permutohedron $\P^3_4$, bounding the region $\A^3_4$, is also a tetrahedron, with vertices $\v{f}^i$ located at the centres of the faces of the original tetrahedron, i.e., $\v{f}^1=\v{v}^3=\frac{1}{3}(1,1,1,0)$ and $\v{f}^i$ for $i>1$ are given by permutations of $\v{f}^1$. However, as we prove in Appendix~\ref{app:tetra_skeleton}, not all points within the tetrahedron $\P^3_4$ belong to $\A_4^3$. More precisely, we show that while the points lying on lines connecting edges of $\P^3_4$ with its centre belong to $\A_4^3$, the faces of $\P^3_4$ and points lying on the lines connecting centres of these faces with the centre of $\P^3_4$ do not belong to $\A^3_4$. We illustrate this in Figs.~\hyperref[fig:tetra_skeleton]{2a-b}. Moreover, based on numerical evidence, we conjecture that points belonging to $\A^3_4$ have the following product structure:
\begin{equation}
	\label{eq:3dist_conjecture}
	a b \v{f}^1+a (1-b) \v{f}^2+(1-a) b \v{f}^3+(1-a) (1-b) \v{f}^4,
\end{equation}
with $a,b\in[0,1]$. We present this conjectured set in Fig.~\hyperref[app:tetra_skeleton]{2c}.

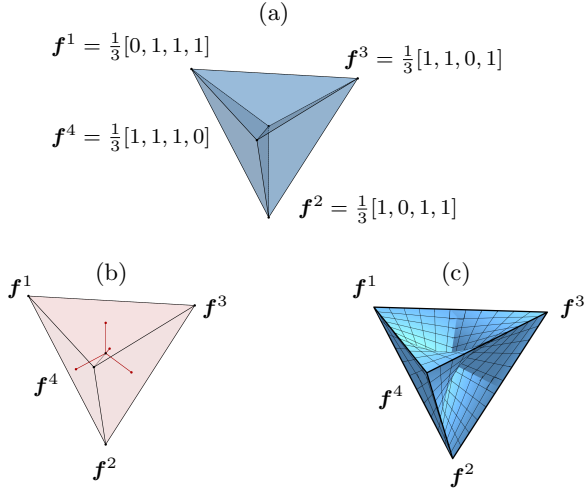
\begin{figure}[t]
	\begin{tikzpicture}			
		\node at (-0.25\columnwidth,0.3\columnwidth) {\includegraphics[width=0.4\columnwidth]{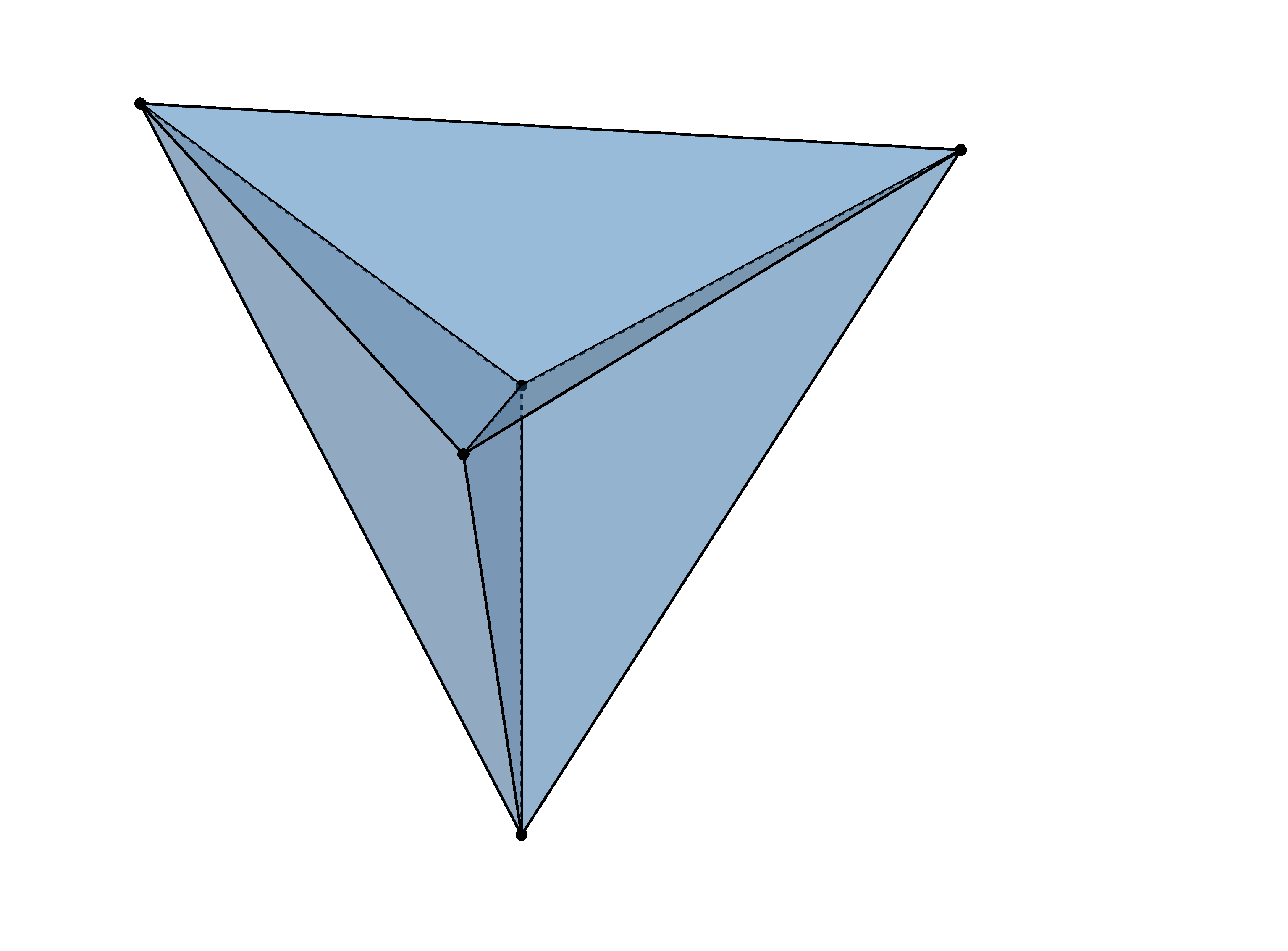}};
		\node at (-0.5\columnwidth,0.45\columnwidth) {\footnotesize\color{black}$\v{f}^1=\frac{1}{3}[0,1,1,1]$};
		\node at (-0.05\columnwidth,0.43\columnwidth)
		{\footnotesize\color{black}$\v{f}^3=\frac{1}{3}[1,1,0,1]$};	
		\node at (-0.5\columnwidth,0.31\columnwidth)
		{\footnotesize\color{black}$\v{f}^4=\frac{1}{3}[1,1,1,0]$};
		\node at (-0.5\columnwidth,-0.05\columnwidth) {\includegraphics[width=0.4\columnwidth]{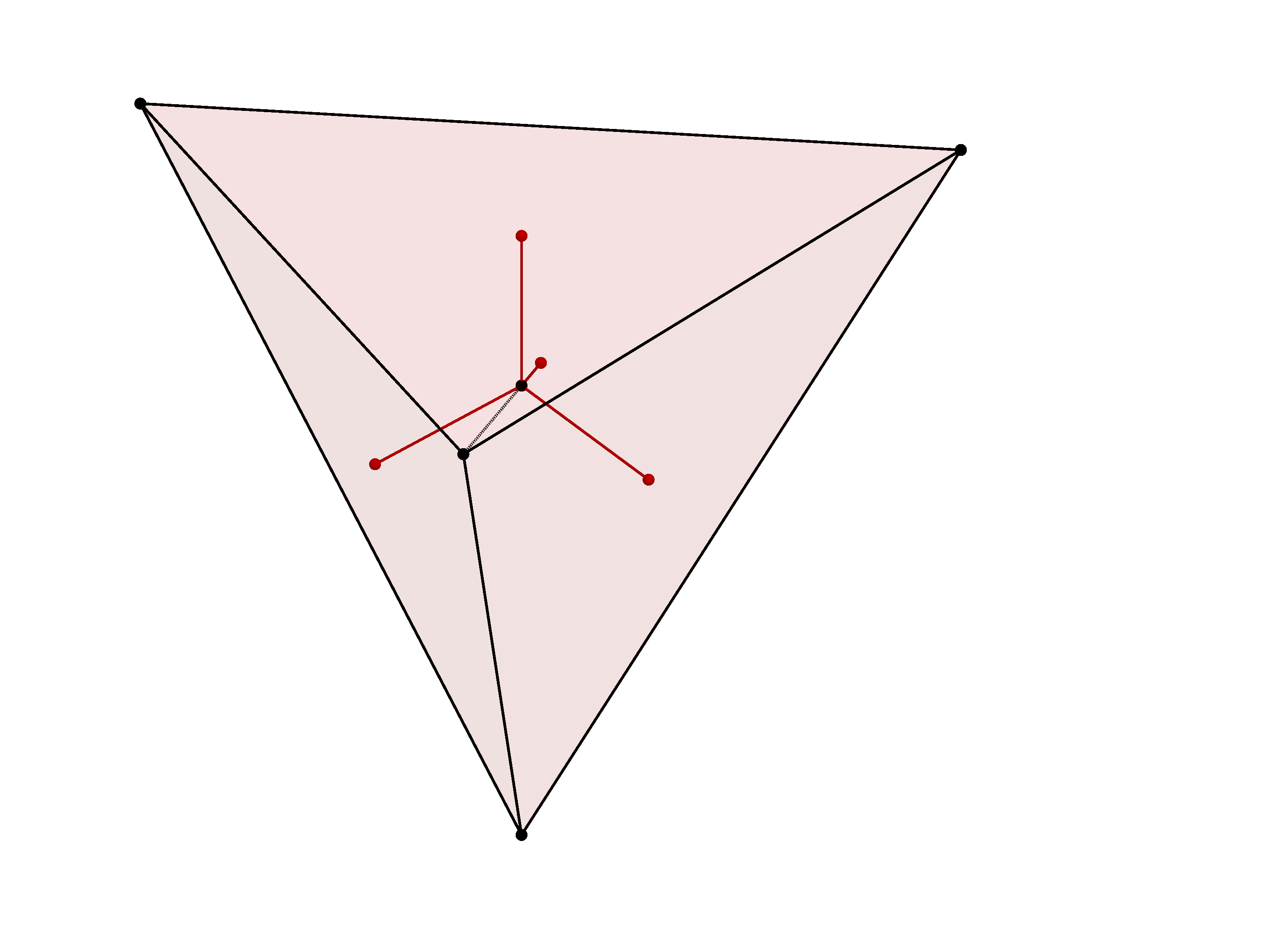}};
		\node at (-0.67\columnwidth,0.08\columnwidth) {\footnotesize\color{black}$\v{f}^1$};
		\node at (-0.37\columnwidth,0.05\columnwidth)
		{\footnotesize\color{black}$\v{f}^3$};	
		\node at (-0.54\columnwidth,-0.2\columnwidth)
		{\footnotesize\color{black}$\v{f}^2$};	
		\node at (-0.63\columnwidth,-0.07\columnwidth)
		{\footnotesize\color{black}$\v{f}^4$};
		\node at (0\columnwidth,-0.05\columnwidth) {\includegraphics[width=0.4\columnwidth]{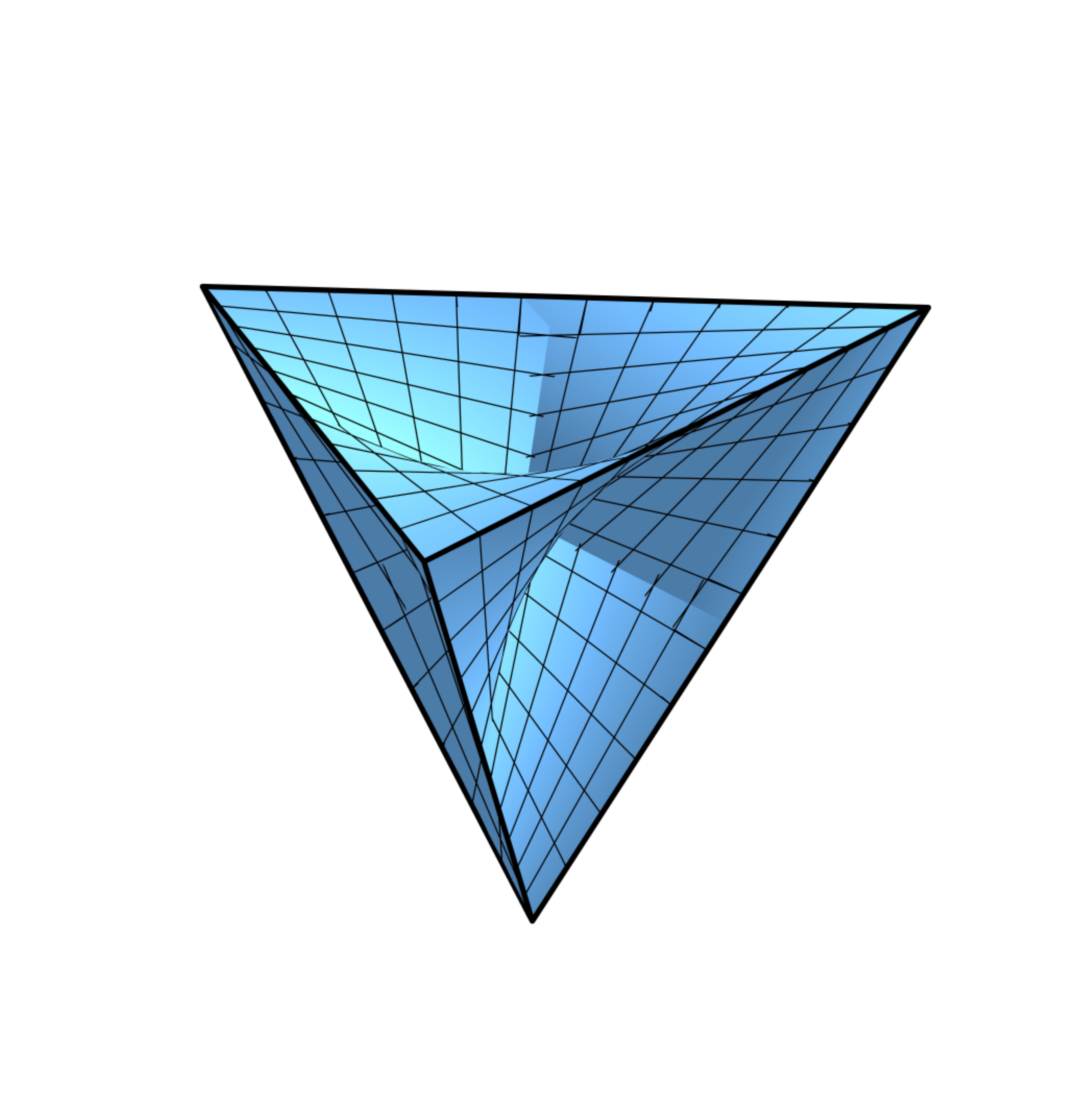}};
		\node at (-0.14\columnwidth,0.08\columnwidth) {\footnotesize\color{black}$\v{f}^1$};
		\node at (0.18\columnwidth,0.05\columnwidth)
		{\footnotesize\color{black}$\v{f}^3$};	
		\node at (0.01\columnwidth,-0.21\columnwidth)
		{\footnotesize\color{black}$\v{f}^2$};	
		\node at (-0.1\columnwidth,-0.09\columnwidth)
		{\footnotesize\color{black}$\v{f}^4$};
		\node at (-0.28\columnwidth,0.5\columnwidth) {\small\color{black} (a)};
		\node at (-0.53\columnwidth,0.1\columnwidth) {\small\color{black}(b)};	
		\node at (0\columnwidth,0.1\columnwidth) {\small\color{black}(c)};	
		\node at (-0.12\columnwidth,0.20\columnwidth)
		{\footnotesize\color{black}$\v{f}^2=\frac{1}{3}[1,0,1,1]$};	
	\end{tikzpicture}
	\caption{\label{fig:tetra_skeleton} \emph{Structure of $\A^3_4$.} (a)~Probability vectors lying on lines connecting edges of $\P^3_4$ with its centre belong to $\A_4^3$. (b)~Probability vectors lying on the faces of $\P^3_4$ and points lying on the lines connecting centres of these faces with the centre of $\P^3_4$ do not belong to $\A^3_4$. (c)~Conjectured form of the distinguishability region $A^3_4$ described by Eq.~\eqref{eq:3dist_conjecture}.}
\end{figure}

\subsection{Properties of $M$-distinguishability regions}

Here, we collect the properties of $M$-distinguishability regions $\A_d^M$ beyond what is stated by the permutohedron bound $\A_d^M\subset \P_d^M$. First, we make two obvious observations: the vertices of $\P_d^M$ always belong $\A_d^M$ and $\A_d^M$ does not have to be convex. The first one comes from the fact that vertices $\Pi_k\v{v}^M$ of $\P_d^M$ correspond to maximally mixed states on the $M$-dimensional subspaces, and we know that then the columns of $M$-dimensional Fourier matrix form $M$ orthogonal states with a fixed classical version $\v{v}^M$. The second observation comes from noting that already $\A_4^3$ is not convex. Despite $M$-distinguishability regions not being convex, we conjecture that they have a related property of being star-shaped.
\begin{conj}
	The $M$-distinguishability regions $\A_d^M$ form star-shaped domains with the centre point given by $\v{\eta}$, i.e.,
	\begin{equation}
		\v{p}\in\A_d^M \quad\Longrightarrow\quad \forall \lambda\in[0,1]:~\lambda\v{p}+(1-\lambda)\v{\eta}\in\A_d^M.
	\end{equation} 
\end{conj}
\begin{rmk}
	The above conjecture, via Lemma~\ref{lem:dist_uni}, is directly related to the known conjecture about the star-shaped property of the set of unistochastic matrices~\cite{bengtsson2005birkhoff}.
\end{rmk}

The next property allows one to conclude that $\v{p}$ belongs to $M$-distinguishability region, if its \emph{coarse-grained} version belongs to it. The definition of coarse-graining and the result are as follows.
\begin{defn}[Coarse-graining]
	The set $\G$ of \emph{coarse-graining} matrices consists of all stochastic matrices with entries in $\{0,1\}$. Moreover, if $\v{q}=G\v{p}$ for some $G\in\G$, then $\v{q}$ is called a coarse-grained version of $\v{p}$.
\end{defn}

\begin{prop}
	If there exist $M$ perfectly distinguishable pure states with a classical version $\v{q}$ given by coarse-graining of $\v{p}$, i.e., $\v{q}=G\v{p}$, then there exists $M$ perfectly distinguishable pure states with a classical version $\v{p}$. 
\end{prop}
\begin{proof}
	Assume that there exists a set of $M$ mutually orthogonal states
	\begin{equation}
		\ket{\psi^{(n)}}=\sum_k \sqrt{q_k}\exp(i\phi^{(n)}_{k})\ket{k},
	\end{equation}
	with $\v{q}=G\v{p}$. Let us denote by $k^*$ the unique index for which \mbox{$G_{k^*k}=1$}, and construct the set of $M$ states
	\begin{equation}
		\ket{\tilde{\psi}^{(n)}}=\sum_k \sqrt{p_k}\exp(i\phi^{(n)}_{k^*})\ket{k}.
	\end{equation}
	These states all have a fixed classical action $\v{p}$ and are mutually orthogonal, 
	\begin{align}
	\bra{\tilde{\psi}^{(m)}}\tilde{\psi}^{(n)}\rangle&=\sum_k p_k \exp(i\phi^{(n)}_{k^*}-i\phi^{(m)}_{k^*})\nonumber\\
	&=\sum_k q_k \exp(i\phi^{(n)}_{k}-i\phi^{(m)}_{k})\nonumber\\
	&=\bra{\psi^{(m)}}\psi^{(n)}\rangle=0,
	\end{align}
	so that $\v{p}\in\A^M_d$.
\end{proof}

In particular, since edges (1-faces) $\v{e}^{(kl)}$ of $\P^M_d$, 
\begin{equation}
	\v{e}^{(kl)}=\lambda \Pi_k\v{v}^M + (1-\lambda) \Pi_l\v{v}^M,\quad \lambda\in[0,1],
\end{equation}
can be coarse-grained to a vertex $\v{v}^M$, they all belong to the distinguishability region $\A^M_d$. On the other hand, faces (2-faces) of $\P^M_d$ do not need to belong to $\A^M_d$, as illustrated in Fig.~\hyperref[fig:tetra_skeleton]{2b}, where no point of the proper 2-face of $\P^3_4$ belongs to $\A_4^3$. In fact, the proof of Proposition~\ref{prop:insufficient} shows that the centres of 2-faces of $\P^{d-1}_d$ do not belong to $\A^{d-1}_d$ for even $d>2$.

Finally, we can say something about the smallest non-trivial distinguishability region $\A_d^{d-1}$. For this, let us first denote by $\B(\v{p},\epsilon)$ the ball of radius $\epsilon$ centred at $\v{p}$, so that $\v{q}\in \B(\v{p},\epsilon)$ if and only if $\delta(\v{p},\v{q})\leq\epsilon$. We then have the following result.

\begin{prop}
	\label{thm:ball_inside}  
	For prime $d$ there exists $d-1$ perfectly distinguishable states with classical version $\v{p}$ if $\v{p}$ is close enough to a maximally mixed distribution, i.e., \mbox{$\delta(\v{p},\v{\eta})\leq\epsilon$} for some $\epsilon>0$. Equivalently, for prime $d$ we have
	\begin{equation}
		\exists {\epsilon > 0}:~ \B(\v{\eta},\epsilon) \subseteq \A^{d-1}_{d}.
	\end{equation}
\end{prop}

\begin{proof}
	Due to Lemma~\ref{lem:dist_uni}, the existence of $d-1$ orthogonal pure states with classical version $\v{p}$, is equivalent to the unistochasticity of the following matrix,
	\begin{equation}
	T(\v{p})=
	\begin{pmatrix}
	p_1 & p_1 & \ldots & p_1 & 1-(d-1) p_1 \\
	p_2 & p_2 & \ldots &p_2 & 1-(d-1) p_2 \\
	\vdots & \vdots& \ddots  &\vdots & \vdots\\
	p_d & p_d & \ldots & p_d & 1-(d-1) p_d 
	\end{pmatrix}.
	\end{equation}
	Now, for $T(\v{p})$ lying in the $\epsilon'$-ball around $T(\v{\eta})=W$, $\v{p}$ lies in an $\epsilon$-ball around $\v{\eta}$. Hence, if we could show that in the $\epsilon'$-ball around $W$ all bistochastic matrices are unistochastic, then we would prove that in the \mbox{$\epsilon$-ball} around $\v{\eta}$ all probability distributions allow for \mbox{$(d-1)$-distinguishability}. Due to result of Ref.~\cite{tadej2008defect}, we know that such a ball exists if there exists an isolated complex Hadamard matrix of dimension $d$. Finally, since it is known that the Fourier matrix $F$ (which is a complex Hadamard matrix) is isolated for any prime $d$, it implies the existence of the postulated ball $\B(\v{\eta},\epsilon) \subset  \A^{d-1}_{d}$ for prime $d$.
\end{proof}

\begin{rmk}
	\label{rmk:no_ball_in4}
	For composite dimensions we know that there exist isolated complex Hadamard matrices for \mbox{$d \in \{5,\ldots,17 \}$}~\cite{bruzda2007complex}. Therefore, the above theorem holds in those dimensions. Moreover, there is no isolated complex Hadamard matrix of order $4$, so there are probabilities infinitesimally close to $\v{\eta}$, which do not allow for $3$-distinguishability (as can be seen in Fig.~\hyperref[fig:tetra_skeleton]{2b} with four directions from the centre having no 3-distinguishable states). 	
\end{rmk}

\section{Classically indistinguishable channels}
\label{sec:channels}

Each classical action $T$ can be represented by specifying $d$ points, each belonging to a distinct simplex $\Delta_d$ and describing the column vectors of $T$. However, since geometrically this picture is not as clear as in the case of classical states, we avoided generalising the concept of distinguishability regions, and instead we focus on distinguishability numbers \mbox{${\M}(T)$} and \mbox{$\tilde{\M}(T)$}. Nevertheless, it is helpful to divide classical actions into three families: unistochastic, bistochastic, and general stochastic matrices. In what follows we describe results concerning each of the families separately, and at the end of the section we also present a full analysis of classically indistinguishable qubit channels. Before we start, notice that for all $T$ we have $\M(T)\leq d^2$, $\tilde{\M}(T)\leq d$ and $\tilde{\M}(T)\leq\M(T)$.

\subsection{Unistochastic action}

We start our study of distinguishability numbers by focusing on channels with a unistochastic classical action~$T$. By definition, there exists at least one unitary channel with a given unistochastic action. In fact, as we now show, for every unistochastic $T$ one can always find $d$ unitary channels that can be perfectly distinguished without using entanglement.
\begin{prop}
	\label{prop:unistochastic}
 	For every unistochastic $T$ the restricted distinguishability number $\tilde{\M}(T)=d$.
\end{prop}
\begin{proof}
	By definition of a unistochastic matrix $T$, there exists a unitary matrix $U$ such that $T = U\circ\bar{U}$. Now, consider a set of $d$ unitary channels $V^{(k)}:=D^{(k)}U$, with $D^{(k)}$ defined in Eq.~\eqref{eq:diagonal} and \mbox{$k\in\{1,\dots,d\}$}. Every such channel has the same classical action given by $T$. Moreover, when acting on a state
 	\begin{equation}
 		\ket{\psi_+}=\frac{1}{\sqrt{d}}U^\dagger\sum_l \ket{l}
 	\end{equation}
 	the set $\{V^{(k)}\}$ produces an orthonormal set of states,
 	\begin{equation}
 		D^{(k)}U\ket{\psi_+}=\frac{1}{\sqrt{d}}\sum_l D^{(k)}\ket{l}=\sum_l F_{kl}\ket{l}=F\ket{k}.
 	\end{equation}
\end{proof}
The obvious next question to ask is whether using entangled input states one can increase this number. As we will show, the answer strongly depends on $T$, with extreme cases given by $T=\iden$ and $T=W$. These correspond to situations where entanglement cannot help at all, and where it raises the number of distinguishable channels all the way to $d^2$. Before proving this statement, let us first introduce a family of Schur-product channels defined by~\cite{li1997special,levick2017quantum}.

\begin{defn}[Schur-product channels]
	The action of a Schur-product channel $\Phi_{X}$ is given by
	\begin{equation}
		\Phi_{X} : \rho \mapsto \rho \circ X,
	\end{equation}
	where the entry-wise product is performed in the distinguished basis and $X$ is an arbitrary correlation matrix, i.e., $X$ is positive and has ones on the diagonal.
\end{defn}

\noindent We can now prove the following result.

\begin{prop}
	\label{prop:schur_product}
	Distinguishability numbers for the identity and van den Waerdan matrix are given by $\M(\iden)=d$ and $\M(W)=d^2$.
\end{prop}
\begin{proof}
	To prove the first part we will show that \mbox{$\tilde{\M}(\iden)\geq \M(\iden)$} which, together with the condition \mbox{$\tilde{\M}(\iden)\leq \M(\iden)$} and Proposition~\ref{prop:unistochastic}, leads to \mbox{$\M(\iden)=d$}. We start by noting that the most general quantum channel consistent with classical action $T=\iden$ is a Schur-product channel (this can be easily seen by comparing their Jamio\l{}kowski states and showing that they are the same). Now, consider $n$ such channels, $\{\Phi^{(n)}\}_{n=1}^M$, each defined via the corresponding correlation matrix $X^{(n)}$. The necessary and sufficient condition for perfect distinguishability between all those channels is the existence of a bipartite state $\ket{\Psi}$ such that that for any two channels, $\Phi^{(m)}$ and $\Phi^{(n)}$, we have~\cite{feng2004unambiguous}
	\begin{equation}
		\label{eq:dist_schur}
		[\Phi^{(m)}\otimes\I (\ketbra{\Psi}{\Psi})][\Phi^{(n)}\otimes\I (\ketbra{\Psi}{\Psi})]=0.
	\end{equation}
	Let us write a general pure bipartite state $\ket{\Psi}$ in the Schmidt basis as
	\begin{equation}
		\ket{\Psi}=\sum_{j=1}^d c_j \ket{a_j b_j},
	\end{equation}
	with
	\begin{equation}
		\ket{a_j}=\sum_{k=1}^d \alpha_{jk} \ket{k},\quad \ket{b_j}=\sum_{k=1}^d \beta_{jk} \ket{k}.
	\end{equation}
	
	Now, by straightforward calculation, one can show that Eq.~\eqref{eq:dist_schur} implies that for all $k,l,p,r$ we have
	\begin{equation}
		\sqrt{c_p c_r} \alpha_{pk} \alpha_{rl}^* \sum_j d_j X^{(m)}_{kj} Y^{(n)}_{jl}=0,
	\end{equation}
	where 
	\begin{equation}
		d_j:=\sum_k |\alpha_{kj}|^2 c_k,
	\end{equation}
	with $d_j\geq 0$ and $\sum_j d_j=1$. Thus, for every $k,l$ we have that either
	\begin{equation}
		\label{eq:dist_schur2}
		\sum_j d_j X^{(m)}_{kj} Y^{(n)}_{jl}=0
	\end{equation}
	or that $\sqrt{c_p}\alpha_{pk}=0$ for every $p$, or that $\sqrt{c_r} \alpha_{rl}^*=0$ for every $r$. The latter two conditions are equivalent to $d_k=0$ and $d_l=0$, which can be seen by squaring and summing the original conditions. We thus conclude that if for some choice of $\ket{\Psi}$ the channels $\{\Phi^{(n)}\}_{n=1}^M$ are all perfectly distinguishable, then for every $k,l$ either $d_k=0$, or $d_l=0$, or Eq.~\eqref{eq:dist_schur2} is true for all $m,n$.
	
	Consider now the action of channels $\{\Phi^{(n)}\}_{n=1}^M$ on a separable state $\ket{\psi}$ defined by
	\begin{equation}
		\ket{\psi}=\sum_j \sqrt{d_j} \ket{j}.
	\end{equation} 
	For the output states to be perfectly distinguishable we need that for every pair $m,n$ the following condition is satisfied~\cite{feng2004unambiguous}
	\begin{equation}
		\Phi^{(m)}(\ketbra{\psi}{\psi})\Phi^{(n)}(\ketbra{\psi}{\psi})=0.
	\end{equation}
	This means that for every $k,l$ we require
	\begin{equation}
		\sqrt{d_k d_l} \sum_j d_j X^{(m)}_{kj} X^{(n)}_{jl}=0.
	\end{equation}
	The above conditions are precisely the same as for distinguishability with arbitrary entangled state, i.e., if there exists an entangled state $\ket{\Psi}$ allowing for perfect distinguishability of all channels $\{\Phi^{(n)}\}_{n=1}^M$, then there also exists a separable state $\ket{\psi}$ allowing for perfect distinguishability. Therefore, as entanglement does not increase the maximal number of distinguishable channels with classical action $T=\iden$, and without entanglement this number is equal to $d$, we conclude that \mbox{$\M(\iden)=d$}.
	
	We now turn to $T=W$ case. Note that all $d^2$ unitary matrices $U^{(kl)}$ defined by
	\begin{equation}
		U^{(kl)}:=D^{(k)} F D^{(l)\dagger},
	\end{equation}
	have the same classical action given by $W$. Moreover, the action of each of these channels on one half of the maximally entangled state $\ket{\Omega}$ produces an orthonormal set of states $\{\ket{\Psi_{kl}}\}$ with
	\begin{equation}
		\ket{\Psi_{kl}} = (U^{(kl)}\otimes\iden)\ket{\Omega}.
	\end{equation}
	To see that are $\{\ket{\Psi_{kl}}\}$ are indeed orthogonal, note that
	\begin{align}
	\langle \Psi_{k'l'} \ket{\Psi_{kl}}&=\frac{1}{d} \sum_{j} \bra{j} D^{(l')} F^\dagger D^{(k')\dagger} D^{(k)} F D^{(l)\dagger}\ket{j}\nonumber\\
	&=\sum_{j,j'} \bra{j} F_{l'j} F^\dagger D^{(k')\dagger}\ketbra{j'}{j'} D^{(k)} F F_{lj}^*\ket{j}\nonumber\\
	&=\sum_{j,j'} F_{l'j} F_{k'j'}^* F_{kj'}  F_{lj}^*=[FF^\dagger]_{l'l} [FF^\dagger]_{kk'}\nonumber\\
	&=\delta_{ll'}\delta_{kk'}.
	\end{align}
	Thus, we conclude that all $d^2$ channels $U^{(kl)}$ with classical action $T=W$ are perfectly distinguishable, and so $\M(W)=d^2$.
\end{proof}

\begin{corol}
	The maximal set of perfectly distinguishable Schur-product channels is $d$.
\end{corol}

For a general matrix $T$ one expects that $\M(T)$ can take all values between the above extremes given by $d$ and~$d^2$. We will now show how the construction used while proving $\M(W)=d^2$ can be generalized, opening a way to construct $N> d$ perfectly distinguishable channels, and thus finding lower bounds on $\M(T)$ for general unistochastic~$T$. First, we restrict our search for the maximal set of perfectly distinguishable channels with a fixed unistochastic classical action $T$ to unitary channels. Since the moduli of every entry for all these unitaries are equal, we may further restrict our considerations to unitaries of the following form
\begin{equation}
	U^{(mn)}=L^{(m)} U R^{(n)},
\end{equation}
where $L^{(m)}$ and $R^{(n)}$ are general diagonal unitaries with 
\begin{equation}
	L^{(m)}_{kk}=\exp\left(i\phi^{(m)}_k\right),\quad R^{(n)}_{kk}=\exp\left(i\theta^{(n)}_k\right),
\end{equation}
and $U$ is any unitary matrix satisfying $U\circ U=T$. Finally, we assume that the input state used in the distinguishability protocol is the maximally entangled state $\ket{\Omega}$. The set of unitaries $\{U^{(mn)}\}$ is then perfectly distinguishable if for all pairs $m,n$ and $m',n'$ we have that the following expression vanishes,
\begin{equation}
	\label{eq:ort_unitaries}
	\!\!\langle\!\langle U^{(mn)}\ket{U^{(m'n')}}\!\rangle=\sum_{k,k'} [L^{(n')\dagger} L^{(n)} T R^{(m)}R^{(m')\dagger}]_{kk'}.\!
\end{equation}
We note the close resemblance of the above problem to pure state distinguishability. There one needed to find phases $\{\phi^{(n)}_{k}\}$, so that \mbox{$\sum_k p_k \exp(i\phi^{(n)}_k-i\phi^{(m)}_{k})$} vanishes for all $m,n$; here, one is looking for phases $\{\phi^{(m)}_k\}$ and $\{\theta^{(n)}_k\}$, so that Eq.~\eqref{eq:ort_unitaries} is satisfied. In Appendix~\ref{app:d+1_unitaries} we show how the above method can be used to find $d+1$ perfectly distinguishable channels with a particular classical action.

\subsection{General stochastic action}

Although for unistochastic action one could always find $d$ perfectly distinguishable channels, it is no longer the case when one considers general stochastic action $T$. In fact, there exist $T$ for which one cannot construct even 2 distinguishable channels. As a particular example consider the completely contractive classical action $T$ defined by $T_{1k}=1$ for all $k$, and $T_{jk}=0$ for all $k$ and $j\neq 1$. This classical action uniquely defines a quantum channel $\Phi(\cdot)=\ketbra{1}{1}$, and thus $\M(T)=1$. 

As channel distinguishability ultimately depends on state distinguishability, we can employ the results from Sec.~\ref{sec:states}.
\begin{prop}
	\label{prop:triangle}
	Consider a classical action $T$ and denote the probability distribution formed from the entries of its $l$-th column by $T_{\star l}$. If, for any $l$, we have $T_{\star l}\in\A^M_d$, then $\tilde{\M}(T)\geq M$.
\end{prop}
\begin{proof}
	Define a set of channels $\{\Phi^{(n)}\}$ by
	\begin{equation}
		\Phi^{(n)}(\rho):=\sum_{l=1}^d K^{(n)}_{l}\rho K^{(n)\dagger}_{l}
	\end{equation}
	with
	\begin{equation}
		K^{(n)}_{l}:=\ket{\psi^{(n)}_l}\bra{l}
	\end{equation}
	and
	\begin{equation}
		\ket{\psi^{(n)}_l}=\sum_{k=1}^d \sqrt{T_{kl}}e^{i\phi^{(n)}_{k}}\ket{k}.
	\end{equation}
	The classical action of each of $\Phi^{(n)}$ is given by $T$ for every choice of phases $\{\phi^{(n)}_{k}\}$,
	\begin{equation}
		\bra{k}{\Phi^{(n)}(\ketbra{l}{l})}\ket{k}=|\langle{k}\ket{\psi^{(n)}_l}|^2=T_{kl}.
	\end{equation}
	At the same time the state $\ket{l}$ is mapped by $\Phi^{(n)}$ to $\ket{\psi^{(n)}_l}$, whose classical version is $T_{\star l}$ independently of $\{\phi^{(n)}_{k}\}$. Therefore, if there exists $M$ perfectly distinguishable states with classical version $T_{\star l}$, then it is possible to choose phases $\{\phi^{(n)}_{k}\}$ so that each $\Phi^{(n)}$ maps $\ket{l}$ to an orthogonal state $\ket{\psi_l^{(n)}}$. 	
\end{proof}
The above result, together with Proposition~\ref{prop:sufficient}, imply the following corollary.
\begin{corol}
	\label{lem:channel_triangle}
	If the entries of at least one column of the stochastic matrix $T$ satisfy the triangle inequality, i.e., the largest entry is smaller than the sum of the remaining entries, then $\tilde{\M}(T)\geq 2$. 
\end{corol}

Even if no column of $T$ satisfies the triangle inequality, there can still exist two perfectly distinguishable channels. This time the distinguishability protocol will require the use of entanglement, but before we state the result, we first need to introduce a particular swap procedure $S_{kl}^{j\alpha}$. Given a stochastic matrix $T$ the matrix $T'=S_{kl}^{j\alpha}(T)$ is obtained by multiplying column $k$ of $T$ by a real number $\alpha$, which is then followed by a transposition of two elements in row $j$, one belonging to column $k$ and the other to column $l$.
\begin{prop}
	\label{prop:transposition}
	Assume that the classical action $T$ can be transformed by some swap procedure $S_{kl}^{j\alpha}$ into a matrix $T'$, such that both columns $k$ and $l$ of $T'$ satisfy the triangle inequality. Then, $\M(T)\geq 2$.
\end{prop}
The proof of the above result can be found in Appendix~\ref{app:transposition}.

\subsection{Bistochastic action}

Finally, we proceed to the results concerning distinguishability of quantum channels with a fixed classical action $T$ that is bistochastic. Our main result states that one can always find at least two perfectly distinguishable channels with a given bistochastic classical action.

\begin{prop}
	For every bistochastic matrix $T$ we have $\M(T)\geq 2$.
\end{prop}

\begin{proof}
	First, if there exists a column of $T$ that satisfies the triangle inequality, i.e., the largest entry is smaller than the sum of the remaining entries, then $\M(T)\geq 2$ due to Lemma~\ref{lem:channel_triangle}. Otherwise, we deal with $T$ such that each column contains an element larger than $\frac{1}{2}$. Without loss of generality we can assume those elements are placed on the diagonal of the matrix $T$. Similarly, without loss of generality we may assume that the largest of the non-diagonal elements of $T$ is $T_{21}$. The plan now is to show that for a proper choice of $\alpha$, the swap procedure $S_{21}^{2\alpha}$ transforms $T$ into $T'$ such that the triangle inequality is satisfied by columns 1 and 2 of $T'$. This will allow us to use Proposition~\ref{prop:transposition} and conclude that $\M(T)\geq 2$. To achieve this we will separately consider two situations: $T_{11}>T_{22}$ and $T_{11} \leq T_{22}$.
	
	First assume $T_{11}<T_{22}$ and choose $S_{21}^{2\alpha}$ with $\alpha =1$, so that
	\begin{equation}
		T =
		\begin{pmatrix}
		T_{11} & T_{12} & \ldots \\
		T_{21}  & T_{22} &\ldots \\
		T_{31} &T_{32}& \ldots \\
		\vdots  & \vdots
		\end{pmatrix}
		\xrightarrow{S_{21}^{2\alpha}}
		T'=
		\begin{pmatrix}
		T_{11} & T_{12} & \ldots \\
		T_{22} & T_{21} &\ldots \\
		T_{31} &T_{32}& \ldots \\
		\vdots  & \vdots
		\end{pmatrix}\!.
	\end{equation}
	As $T$ is bistochastic, we have
	\begin{align}
			T_{21} + \sum_{k \neq 1} T_{2k} = 1 &\implies  T_{21} \leq 1 - T_{22}=\sum_{k \neq 2} T_{k2},
	\end{align}
	therefore the second column of $T'$ satisfies the triangle inequality (because the largest element in column 2 of $T'$ is $T_{21}$). Similarly, we have
	\begin{equation}
		T_{22}\leq 1-T_{21}=\sum_{k \neq 2} T_{k1},
	\end{equation}
	so the first column of $T'$ also satisfies the triangle inequality (because the largest element in column 1 of $T'$ is $T_{22}$). We conclude that, due to Proposition~\ref{prop:transposition}, there exist two perfectly distinguishable channels with classical action $T$.
	
	Let us now turn to the second case, $T_{11}\geq T_{22}$. Again, we obtain $T'$ by a swap procedure $S_{21}^{2\alpha}$,
	\begin{equation}
		T=
		\begin{pmatrix}
			T_{11} & T_{12}& \ldots \\
			T_{21}  & T_{22} &\ldots \\
			T_{31} &T_{32}& \ldots \\
			\vdots  & \vdots
		\end{pmatrix}
		\rightarrow
		T'=
		\begin{pmatrix}
			T_{11} & \alpha T_{12}& \ldots \\
			\alpha T_{22}  & T_{21} &\ldots \\
			T_{31} & \alpha T_{32}& \ldots \\
			\vdots  & \vdots
		\end{pmatrix}
	\end{equation}
	with
	\begin{equation}
		\alpha = \frac{2 T_{11} + T_{21} -1}{T_{22}}.
	\end{equation}
	Due to bistochasticity of $T$ we have $T_{21}\leq 1-T_{11}$, so
	\begin{align}
		\alpha T_{22}  = 2T_{11} + T_{21}-1 \leq T_{11},
	\end{align}
	meaning that $T_{11}$ is the largest element in the first column of $T'$. This column satisfies the triangle inequality, because
	\begin{align}
		\sum_{k \neq 1} T'_{k1} &= \alpha T_{22} + (1 - T_{11} -T_{21})=T_{11}.
	\end{align}
	
	It remains to show that the second column of $T'$ satisfies the triangle inequality. We denote the second largest element in the second column of $T$ by $x$. If $T_{21} \geq \alpha x$ then $T_{21}$ is the greatest element in the second column of~$T'$. Then, the triangle inequality has the following form
	\begin{align}
		T_{21} \leq \sum_{k \neq 2} \alpha T_{k2}= \alpha (1- T_{22}),
	\end{align}
	which is equivalent to
	\begin{equation}
		(2T_{11} - 1)(1-T_{22}) - (2T_{22} - 1)T_{21}\geq 0.
	\end{equation}
	As $2T_{11} -1 \geq 2T_{22} -1$ and $1-T_{22} \geq T_{21}$, the above inequality holds.
	
	If $T_{21} < \alpha x$ then $\alpha x$ is the greatest element in the second column of~$T'$. Then, the triangle inequality has the following form
	\begin{equation}
		\alpha x \leq \sum_{k \neq 2} \alpha T_{k2} - \alpha x + T_{21} =\alpha(1-T_{22}-x)+T_{21},\!
	\end{equation}
	which is equivalent to
	\begin{equation}
		\alpha(1-T_{22})-2\alpha x + T_{21}\geq 0.
	\end{equation}
	Since $T_{21} \geq x$ it is sufficient to check that the function $f$ defined by	
	\begin{equation}
		f(T_{21}) := \alpha(1-T_{22}) - T_{21}(2\alpha -1)
	\end{equation}
	is greater or equal $0$ for $T_{21} \in  [0,1-T_{11}]$. The function $f$ is concave, i.e., 
	\begin{equation}
		f(ps + (1-p)t) \geq p f(s) + (1-p)f(t),
	\end{equation}
	with $p \in [0,1]$. It is thus sufficient to check that $f(0) \geq 0$ and $f(1-T_{11}) \geq 0$. By straightforward calculation, one obtains
	\begin{subequations}
	\begin{align}
		f(0) &= \frac{1}{T_{22}}(2T_{11}-1)(1-T_{22})\geq 0, \\
		f(1-T_{11}) &= \frac{1}{T_{22}}(2T_{11}-1)(T_{11}-T_{22})\geq 0.
	\end{align}
	\end{subequations}
	Therefore $f(T_{21})\geq 0$ for all $T_{21} \in [0,1-T_{11}]$. This means that the triangle inequality is satisfied for the second column of $T'$ and, due to Proposition~\ref{prop:transposition}, ends the proof.
\end{proof}

The above result can be further refined for a particular subset of bistochastic matrices defined in the following way.
\begin{defn}[Circulant matrix]
	A stochastic matrix $T$ is called \textit{circulant} if it is of the form:
	\begin{equation}
		\label{eq:circulant}
		T = \sum_{k=1}^{d} \lambda_k X^k,
	\end{equation}
	with $\v{\lambda}\in\Delta_d$ and 
	\begin{equation}
		X=\sum_{k=1}^{d} \ketbra{k\oplus 1}{k},
	\end{equation}
	where $\oplus$ denotes addition modulo $d$.
\end{defn}

For this particular family of bistochastic matrices we can prove a result analogous to Proposition~\ref{prop:unistochastic} that concerns unistochastic matrices.
\begin{prop}
	\label{prop:circulant}
	For every circulant $T$ the restricted distinguishability number $\tilde{\M}(T)=d$.
\end{prop}

\begin{proof}
	For a given circulant matrix $T$, define a set of $d$ quantum channels \mbox{$\{\Phi^{(n)}\}_{n=1}^d$} through their Jamio{\l}kowski states,
	\begin{equation}
		J_{\Phi^{(n)}} = \sum_{\alpha=1}^{d} \lambda_\alpha \ket{\psi_\alpha^{(n)}}\bra{\psi_\alpha^{(n)}},
	\end{equation}
	with $\v{\lambda}$ defining $T$ through Eq.~\eqref{eq:circulant} and $\ket{\psi_\alpha^{(n)}}$ given by	
	\begin{equation}
		\ket{\psi_\alpha^{(n)}}=\frac{1}{\sqrt{d}} \sum_{k=1}^{d} F_{nk}\ket{k,k \oplus \alpha}.
	\end{equation}
	By direct inspection one can check that 
	\begin{equation}
		\bra{kl}J_{\Phi^{(n)}}\ket{kl}=\frac{1}{d}T_{kl},
	\end{equation}
	so that for all $n$ the classical action of $\Phi^{(n)}$ is given by $T$. Moreover,	\begin{equation}
		\bra{\psi_\beta^{(m)}}{\psi_{\alpha}^{(n)} }\rangle =\delta_{\alpha\beta} \delta_{mn} ,
	\end{equation}
	with the first Kronecker delta coming from orthogonality of supports and the second one from orthogonality of columns of the Fourier matrix $F$. This implies orthogonality of the Jamio{\l}kowski states $J_{\Phi^{(n)}}$, and thus quantum channels $\Phi^{(n)}$ sharing the same classical action $T$ are perfectly distinguishable.
\end{proof}

\subsection{Qubit channels}

In this final section we provide a solution for the problem of distinguishing classically indistinguishable qubit channels. We start by noting that a classical action of a general qubit channel is given by
\begin{equation}
	\label{eq:qubit_classical}
	T=\begin{pmatrix}
		a & 1-b\\
		1-a & b
	\end{pmatrix},
\end{equation}
with $0\leq a,b\leq 1$. For $a=b$ we deal with bistochastic matrices that for two-dimensional systems coincide with unistochastic matrices. Without loss of generality, we may assume that $a \geq b $ and introduce \mbox{$\Delta:=a - b\geq 0$}.	

The full characterization of restricted distinguishability numbers for qubit channels is given by the following Proposition and is illustrated in Fig.~\hyperref[fig:qubit]{3a}.

\begin{prop}
	\label{prop:qubit_separable}
 	Restricted distinguishability number $\tilde{\M}(T)$  for a qubit classical action $T$ parametrized as in Eq.~\eqref{eq:qubit_classical} is given by
	\begin{equation}
		\tilde{\M}(T)=\begin{dcases}
		 	1&:\mathrm{~~for~~}\frac{1}{2}<|a-b|\leq 1,\\
		 	2&:\mathrm{~~for~~}0\leq |a-b|\leq \frac{1}{2}.\\
		\end{dcases}
	\end{equation}
\end{prop}

\begin{proof}
	We first consider the case \mbox{$\Delta \in \left[ 0 , \frac{1}{2} \right] $}. We define two quantum channels, $\Phi^{(+)}$ and $\Phi^{(-)}$, in the following way,
	\begin{equation}
		\Phi^{(\pm)}(\cdot)= \Psi^{(\pm)}(U(\cdot)U^\dagger),
	\end{equation}
	with a unitary $U$ given by
	\begin{equation}
		U = \frac{1}{\sqrt{1-\Delta}}
		\begin{pmatrix}
		\sqrt{1-a} & -\sqrt{b}\\
		\sqrt{b} & \sqrt{1-a}
		\end{pmatrix}
	\end{equation}
	and quantum channels $\Psi^{(\pm)}$ defined by their Jamio{\l}kowski states,
	\begin{equation}
		J_{\Psi^{(\pm)}} = \frac{1}{2}
		\begin{pmatrix}
		\Delta & 0 & \pm \Delta & 0\\
		0 & 1 &\pm\sqrt{1 - 2\Delta} &0 \\
		\pm \Delta & \pm\sqrt{1 - 2\Delta} & 1 - \Delta &0\\
		0&0&0&0
		\end{pmatrix}.
	\end{equation}
	Using the fact that the Jamio{\l}kowski states of $\Phi^{(\pm)}$ are related to those of $\Psi^{(\pm)}$ by
	\begin{equation}
		J_{\Phi^{(\pm)}}=\left(\1 \otimes U \right) J_{\Psi^{(\pm)}} \left(\1 \otimes U^{\dagger} \right),
	\end{equation}
	it is straightforward to verify that the classical action of $\Phi^{(\pm)}$ (encoded on the diagonal of $J_{\Phi^{(\pm)}}$) is given by $T$ parametrized as in Eq.~\eqref{eq:qubit_classical}. Moreover, a state \mbox{$\rho= U^\dagger \ketbra{\psi}{\psi} U$} with
	\begin{equation}
		\ket{\psi} = \frac{\ket{0}+\sqrt{1 - 2 \Delta}\ket{1}}{\sqrt{2(1-\Delta)}}
	\end{equation}
	is mapped by $\Phi^{(\pm)}$ to orthogonal states,
	\begin{equation}
		\Phi^{(\pm)}(\rho)=\Psi^{(\pm)}(\ketbra{\psi}{\psi})=\ketbra{\pm}{\pm},
	\end{equation}
	so that $\Phi^{(+)}$ and $\Phi^{(-)}$ are perfectly distinguishable.

	\begin{figure}[t]
		\includegraphics[width=\columnwidth]{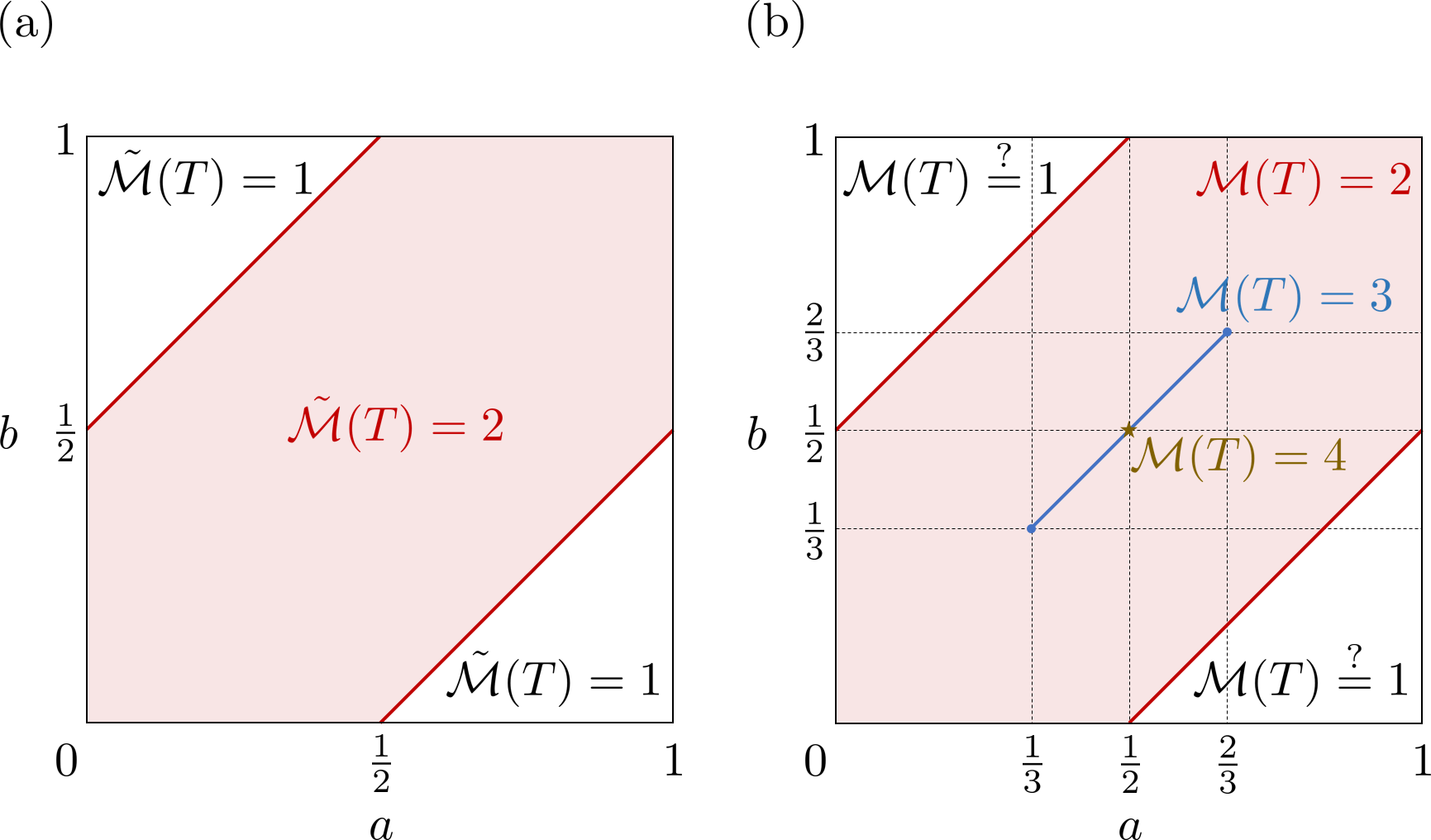}
		\caption{\label{fig:qubit} \emph{Distinguishing qubit channels.} Regions of the parameter space, describing the set of qubit classical actions via Eq.~\eqref{eq:qubit_classical}, corresponding to different distinguishability numbers. (a) Regions with different restricted distinguishability numbers $\tilde{\M}(T)$. (b) Regions with different distinguishability numbers $\M(T)$. The star in the middle corresponds to van der Waerden matrix $W$ that can be coherified to four perfectly distinguishable unitary matrices given by Eqs.~\eqref{eq:van_der_coher1}-\eqref{eq:van_der_coher2}; while the endpoints of the blue segment can be coherified to three perfectly distinguishable unitary matrices given by Eqs.~\eqref{eq:unitary_endpoint1}-\eqref{eq:unitary_endpoint3} with $\phi=\theta=2\pi/3$ or $\phi=4\pi/3$ and $\theta=2\pi/3$.}
	\end{figure}
	
	We will now show that for \mbox{$\Delta >\frac{1}{2}$} we have \mbox{$\M(T)=1$}. The image of a qubit channel $\Phi$ is an ellipsoid inside a Bloch ball. Antipodal points $\ketbra{0}{0}$
	and $\ketbra{1}{1}$ are mapped onto antipodal points,
	\begin{subequations}
	\begin{align}
		\Phi\left( \ketbra{0}{0} \right) & = \frac{1}{2}
		\begin{pmatrix}
			1 + z_{0} & x_{0} - i y_{0} \\
			x_{0} + i y_{0} & 1 - z_{0}
		\end{pmatrix}, \\
		\Phi\left( \ketbra{1}{1} \right) & = \frac{1}{2}
		\begin{pmatrix}
			1 + z_{1} & x_{1} - i y_{1} \\
			x_{1} + i y_{1} &1 - z_{1}
		\end{pmatrix}.
	\end{align}
	\end{subequations}	
	Fixing $T$ corresponds to fixing $z_0$ and $z_1$,
	\begin{subequations}
	\begin{align}
		\bra{0} \Phi\left( \ketbra{0}{0} \right) \ket{0} &= \frac{1}{2} \left( 1 + z_{0} \right)=a, \\
		\bra{1} \Phi\left( \ketbra{1}{1} \right) \ket{1} &= \frac{1}{2} \left( 1 - z_{1} \right)=b,
	\end{align}
	\end{subequations}
	so that $z_{0} = 2a- 1$ and	$z_{1} = 1 - 2b$. Now, the centre of the ellipsoid lies in the middle between the antipodal points and
	its $z$ coordinate is equal to \mbox{$z_{c} = \Delta$}. The $z$ coordinate of any point belonging to ellipsoid must lie between $z_{c}  + \zeta$ and $z_{c} - \zeta$, for some \mbox{$\zeta\geq 0$}. If \mbox{$z_c=\Delta > \frac{1}{2}$} then $\zeta < \frac{1}{2}$,
	because the ellipsoid has to lie inside the Bloch ball. Thus, \mbox{$z_{c} - \zeta > 0$} and
	the entire ellipsoid lies inside the northern hemisphere of the Bloch ball. Analogously, if \mbox{$z_c=\Delta < - \frac{1}{2}$} then  $\zeta  < \frac{1}{2}$ and the entire ellipsoid lies inside the southern hemisphere. In either case, regardless of the choice of $\Phi$, the image of $\Phi$	lies entirely inside one of the hemispheres and does not contain two orthogonal states. This implies that one cannot construct two channels with the same classical action $T$ that will be perfectly distinguishable without using entangled states.
\end{proof}

We now proceed to entangled-assisted distinguishability protocols. Our results on distinguishability numbers for qubit channels are captured by the following Proposition and are illustrated in Fig.~\hyperref[fig:qubit]{3b}

\begin{prop}
	Distinguishability number $\M(T)$ for a qubit classical action $T$ parametrized as in Eq.~\eqref{eq:qubit_classical} satisfies
	\begin{equation}
		\!\!\!\M(T)=\begin{dcases}
		2&:\mathrm{~for~}0<|a-b|\leq \frac{1}{2},\\
		3&:\mathrm{~for~}a=b\mathrm{~and~} a\in\left[\frac{1}{3},\frac{2}{3}\right]\setminus\left\{\frac{1}{2}\right\},\\
		4&:\mathrm{~for~} a=b=\frac{1}{2}.	
		\end{dcases}
	\end{equation}
\end{prop}

\begin{proof}
	First, we will focus on classical action $T$ that is not bistochastic, i.e., $a\neq b$ implying $\Delta>0$. A general entangled two-qubit state is given by
	\begin{equation}
		\ket{\psi}=\sum_{j,k=0}^1 c_{jk} \ket{jk}.
	\end{equation}
	Now, the output $\rho^{(n)}$ of a channel $\Phi^{(n)}$ acting on one part of this state can be written as
	\begin{equation}
		\label{eq:rank_argument1}
		\!\!\rho^{(n)}:=\Phi^{(n)}\otimes\I (\ketbra{\psi}{\psi})=(\iden\otimes[\psi]^\top)J_{\Phi^{(n)}} (\iden\otimes[\psi]^\top)^\dagger,\!
	\end{equation}
	where $[\psi]$ is a matrix obtained from $\ket{\psi}$ via mapping $\ket{jk}\rightarrow\ketbra{j}{k}$. Since we analyse entanglement-assisted discrimination, we may assume that the Schmidt number of $\ket{\psi}$ is 2, and thus $[\psi]$ is invertible. Moreover, as $T$ is not unistochastic, $\Phi^{(n)}$ cannot be a unitary~\cite{korzekwa2018coherifying} and thus the rank of $J_{\Phi^{(n)}}$ has to be at least 2. Therefore, for any channel $\Phi^{(n)}$ whose classical action $T$ is not bistochastic, the rank of the output state $\rho^{(n)}$ must be at least $2$. However, the necessary condition for the set $\{\rho^{(n)}\}_{n=1}^M$ to be mutually orthogonal is
	\begin{equation}
		\label{eq:rank_argument2}
		\sum_{n=1}^M \mathrm{rank}(\rho^{(n)})\leq 4,
	\end{equation}
	and so $M\leq 2$. This proves that for $0<|a-b|\leq 1/2$ we have $\M(T)=2$, due to
	\begin{equation}
		2=\tilde{\M}(T)\leq \M(T) \leq 2,
	\end{equation}
	where the first equality comes from Proposition~\ref{prop:qubit_separable}.
	
	We now proceed to bistochastic classical actions, \mbox{$a=b$}. We will first show that if two qubit unitary channels, $U$ and $V$, are perfectly distinguishable with the use of some entangled state $\ket{\Psi}$, they are mutually orthogonal, $\tr{UV^\dagger}=0$, meaning also that they are perfectly distinguishable with the use of maximally entangled state $\ket{\Omega}$. As a result, looking for $M$ perfectly distinguishable unitary channels, we may only focus on a single input state $\ket{\Omega}$. To see this, note that perfect distinguishability of unitaries $U$ and $V$ with the use of state $\ket{\Psi}$ means
	\begin{align}
		0 &= \bra{\Psi} \left( V^{\dagger} \otimes \1 \right) \left( U \otimes \1 \right)  \ket{\Psi}	 = \tr { V^{\dagger} U \rho},
	\end{align}
	where $\rho = \trr{2}{\ketbra{\Psi}{\Psi}}$. Since the matrix $V^{\dagger} U$ has a spectral decomposition
	\begin{equation}
		V^{\dagger} U   =
		U_{0}
		\begin{pmatrix}
		e^{i \phi_1} & 0 \\
		0 & e^{i \phi_2}
		\end{pmatrix}
		U_{0}^{\dagger} = U \Lambda U_0^{\dagger},
	\end{equation}
	we can rewrite the distinguishability condition as
	\begin{equation}
		\tr { V^{\dagger} U \rho} =\tr{ \Lambda U_{0}^{\dagger} \rho U_0 } = 0.
	\end{equation}
	Defining $\rho^{\prime} =  U_{0}^{\dagger} \rho U_0 $, we obtain
	\begin{equation}
		e^{i \phi_1} \rho^{\prime}_{11} + e^{i \phi_2} \rho^{\prime}_{22} = 0.
	\end{equation}
	This implies $\rho^{\prime}_{11} =\rho^{\prime}_{22}$ and $e^{i \phi_1} = - e^{i \phi_2}  $. Thus,
	\begin{equation}
		\tr { V^{\dagger} U } = \tr{\Lambda} = 0.
	\end{equation}
	But, this implies perfect distinguishability between $U$ and $V$ with the use of maximally entangled state $\ket{\Omega}$, because 
	\begin{align*}
		\bra{\Omega} \left( V^{\dagger} \otimes \1 \right) \left( U \otimes \1 \right)  \ket{\Omega} = \tr{V^{\dagger}{U}}=0.
	\end{align*}
	
	We will now find necessary conditions for $\M(T)=4$. Using an analogous rank argument as before (captured by Eqs.~\eqref{eq:rank_argument1}-\eqref{eq:rank_argument2}), we see that all four channels must be rank~1, i.e., be unitary. As explained above, these unitaries $U_i$ must be orthogonal, meaning that there must exist four mutually orthogonal states $\frac{1}{\sqrt{2}}|U_i\rangle\rangle$ with classical version $\frac{1}{2}(a,1-a,1-a,a)$. From the permutohedron bound, Proposition~\ref{prop:necessary}, we know that a necessary condition for this is $a=1/2$. Moreover, this condition is sufficient, as the following four unitaries with classical action $T$ are all mutually orthogonal:
	\begin{subequations}
	\begin{align}
		\label{eq:van_der_coher1}
		U_1=\frac{1}{\sqrt{2}}\begin{pmatrix}
		-1&1\\
		1&1
		\end{pmatrix},\quad
		U_2=\frac{1}{\sqrt{2}}\begin{pmatrix}
		1&-1\\
		1&1
		\end{pmatrix},\\
		U_3=\frac{1}{\sqrt{2}}\begin{pmatrix}
		1&1\\
		-1&1
		\end{pmatrix},\quad
		U_4=\frac{1}{\sqrt{2}}\begin{pmatrix}
		1&1\\
		1&-1
		\end{pmatrix}.
		\label{eq:van_der_coher2}
	\end{align}
	\end{subequations}
	
	We proceed to finding necessary conditions for $\M(T)=3$. Again, from the rank argument, the considered three channels are either all unitary, or two of them are unitary and one has rank 2. In the first case, we can use the orthogonality condition, so that the existence of three perfectly distinguishable unitary channels is equivalent to the existence of three mutually orthogonal states $\frac{1}{\sqrt{2}}|U_i\rangle\rangle$ with classical version $\frac{1}{2}(a,1-a,1-a,a)$. From the permutohedron bound, we clearly see that it is possible only if $a\in[1/3,2/3]$. In the second case, we have two unitary channels $U$ and $V$, and the third channel is a mixed unitary channel
	\begin{equation}
		\label{eq:mixed_unitary}
		\Phi(\cdot)=\lambda S(\cdot)S^\dagger + (1-\lambda) T(\cdot)T^\dagger,
	\end{equation}
	with $S,T$ unitary and $\lambda\in(0,1)$, because all unital (bistochastic) qubit channels are mixed-unitary channels~\cite{landau1993birkhoff}. Perfect distinguishability between $U$ and $\Phi$ implies then that one can perfectly distinguish between $U$ and $S$, and between $U$ and $T$. Analogous implication holds for $V$. Therefore, perfect distinguishability between $U$, $V$ and $\Phi$ is equivalent to the existence of two sets of mutually orthogonal vectors: \mbox{$\frac{1}{\sqrt{2}}\{\vect{U},\vect{V},\vect{S}\}$}
	and  \mbox{$\frac{1}{\sqrt{2}}\{\vect{U},\vect{V},\vect{T}\}$}. Applying the permutohedron bound to these two sets yields:
	\begin{subequations}
	\begin{align}
		\!\!\! (a,1-a,1-a,a)+\frac{1}{2}(s_1,s_2,s_3,s_4)\leq (1,1,1,1)
		\label{eq:perm_add_1},\\
		\!\!\! (a,1-a,1-a,a)+\frac{1}{2}(t_1,t_2,t_3,t_4)\leq (1,1,1,1),
		\label{eq:perm_add_2}
	\end{align}
	\end{subequations}
	where the inequalities are elementwise and vectors $\v{s}$ and $\v{t}$, according to Eq.~\eqref{eq:mixed_unitary}, satisfy
	\begin{equation}
		\lambda \v{s} + (1-\lambda) \v{t} =(a,1-a,1-a,a).
	\end{equation}
	We clearly see that if $a>2/3$ then either Eq.~\eqref{eq:perm_add_1} or Eq.~\eqref{eq:perm_add_1} does not hold, because either $s_1$ or $t_1$ must be larger than $a$. Similarly, one of these equations does not hold for $a<1/3$, because either $s_2$ or $t_2$ must be larger than $1-a$.
	
	We can thus conclude that the necessary condition for $\M(T)=3$ is $a\in[1/3,2/3]$. Moreover, this condition is sufficient, since one can find three unitary channels with a fixed classical action $T$ that, when acting on one part of a maximally entangled state $\ket{\Omega}$, map it to three orthogonal states. More precisely, consider the following unitaries	
	\begin{subequations}
	\begin{align}
		\label{eq:unitary_endpoint1}
		U_{1} &=
		\begin{pmatrix}
		\sqrt{a} & \sqrt{1-a} \\
		\sqrt{1-a} & - \sqrt{a}
		\end{pmatrix}, \\
		\label{eq:unitary_endpoint2}
		U_{2} &=
		\begin{pmatrix}
		\sqrt{a} & \sqrt{1-a} e^{i \theta} \\
		\sqrt{1-a} e^{i \phi} & -\sqrt{a} e^{i (\phi+\theta)}
		\end{pmatrix}, \\
		\label{eq:unitary_endpoint3}
		U_{3} &=
		\begin{pmatrix}
		\sqrt{a} & \sqrt{1-a} e^{i 2\theta} \\
		\sqrt{1-a} e^{2 i \phi} & -\sqrt{a} e^{2 i (\phi+\theta)}
		\end{pmatrix}.
	\end{align}
	\end{subequations}
	with $\phi$ and $\theta$ specified by:
	\begin{equation}
		2a-1 =
		\left(\cot \frac{\phi}{2} \right) \sqrt{ \frac{1 - \cot^2 \frac{\phi}{2}}{1 + 3\cot^2 \frac{\phi}{2}} },
	\end{equation}
	and
	\begin{equation}
		\cot^2 \frac{\theta}{2} = \frac{1 - \cot^2 \frac{\phi}{2}}{1 + 3\cot^2 \frac{\phi}{2}}.
	\end{equation}	
	One can check by direct calculation that when \mbox{$a\in[1/3,2/3]$} the above unitaries are indeed orthogonal.
\end{proof}

As a final remark, let us comment on the most well-studied qubit channels, the phase-damping channel and the (generalised) amplitude-damping channel, from the perspective of our work. It is straightforward to notice that the classical action of a phase-damping channel, specified by Kraus operators
\begin{equation}
\label{eq:phase_damp}
	K^{PD}_1=\begin{pmatrix}
		1 & 0\\
		0 & \sqrt{1-\lambda}
	\end{pmatrix},\quad
	K^{PD}_2=\begin{pmatrix}
		0 & 0\\
		0 & \sqrt{\lambda}
	\end{pmatrix},
\end{equation}
is given by the identity matrix for any value of the damping parameter $\lambda\in[0,1]$. Thus, through Proposition~\ref{prop:schur_product}, there exist two perfectly distinguishable channels with the same classical action as the phase-damping channel, e.g., the identity and the phase flip channels. On the other hand, the classical action of the amplitude-damping channel, specified by Kraus operators
\begin{equation}
\label{eq:amplitude_damp}
K^{AD}_1=\begin{pmatrix}
	1 & 0\\
	0 & \sqrt{1-\gamma}
\end{pmatrix},\quad
K^{AD}_2=\begin{pmatrix}
	0 & \sqrt{\gamma}\\
	0 & 0
\end{pmatrix},
\end{equation}
is given by the matrix $T$ from Eq.~\eqref{eq:qubit_classical}, with $a=1$ and $b=1-\gamma$. From Figs.~\hyperref[fig:qubit]{3a}~and~\hyperref[fig:qubit]{3b}, we see that for small damping parameters $\gamma\leq 1/2$ there are two perfectly distinguishable channels with the same classical action as the amplitude damping channel, but for $\gamma>1/2$ there exists only one such channel. Finally, the classical action of the generalised amplitude damping channel, specified by Kraus operators
\begin{subequations}
\begin{align}
\label{eq:gen_amplitude_damp}
&K^{GAD}_1=\sqrt{p}K_1^{AD},\quad
K^{GAD}_2=\sqrt{p}K_2^{AD},\\
&K^{GAD}_3=\sqrt{1-p}\begin{pmatrix}
\sqrt{1-\gamma} & 0\\
0 & 1
\end{pmatrix},\\
&K^{GAD}_4=\sqrt{1-p}K_2^{AD\dagger},
\end{align}
\end{subequations}
is given by $T$ with \mbox{$a=p+(1-p)(1-\gamma)$} and \mbox{$b=1-p+p(1-\gamma)$}. This means that, depending on $p\in[0,1]$ and $\gamma\in[0,1]$, the parameters of the classical action can take the values $a\in[0,1]$ and $b\in[1-a,1]$ (the upper-right half of Figs.~\hyperref[fig:qubit]{3a}~and~\hyperref[fig:qubit]{3b}). Therefore, the number of perfectly distinguishable channels with the same classical action as the generalised amplitude-damping channel can vary between 1 (e.g., for $p=1$ and $\gamma=1$) and 4 (only for $p=1/2$ and $\gamma=1$).

\section{Outlook}
\label{sec:outlook}

In this work, motivated by the studies on loss of quantum information due to decoherence, we analyzed different ways in which one can coherify a classical probability vector to obtain distinct quantum states. More precisely, we investigated the problem of finding the maximal number of perfectly distinguishable quantum states which all decohere to the same classical state represented by a fixed probability vector. We described general properties and found bounds for the $M$-distinguishability regions ${\cal A}_d^M$ -- the subsets of the probability simplex containing classical states that can be coherified to $M$ perfectly distinguishable quantum states. 

An analogous problem was studied for classical stochastic matrices, which can be coherified into quantum channels. For a given stochastic transition matrix $T$ of order $d$ we studied the distinguishability number ${\cal M}(T)$ and the restricted distinguishability number $\tilde{\M}(T)$ -- the maximal number of perfectly distinguishable quantum channels (with and without the access to entangled states) which share the same classical action $T$. We found general bounds for distinguishability numbers, showed that $\tilde{\M}(T)=d$ for all unistochastic $T$ and that $\M(T)\geq 2$ for all bistochastic $T$. We have also solved this problem in the simplest case of $d=2$, characterizing the set of classically indistinguishable qubit channels.

Our work opens many potential avenues for future research. First, in the current work we have focused exclusively on the condition of perfect distinguishability, so a natural next question concerns the behaviour of $M$-distinguishability regions (and distinguishability numbers) under $\epsilon$-smoothing of that condition, i.e., when a distinguishability protocol is allowed to fail with some small probability $\epsilon$. This is not only important from a practical point of view (as in any realistic protocol state preparations are prone to noise), but may also bring deeper insight into the structure of the sets of classically indistinguishable states and channels. Note, for example, that for large prime dimensions $d$, while there are obviously just $d$ orthogonal states, one can construct $d^2+d$ \emph{almost orthogonal} states, with overlap $1/d\xrightarrow{d\rightarrow\infty} 0$, by choosing $d$ basis states from each of $d+1$ mutually unbiased bases~\cite{durt2010mutually}. This suggests that $\epsilon$-smoothing might have a significant effect on $M$-distinguishability regions and a particular technical question one may want to ask is: how does the error $\epsilon$ of distinguishing $M$ states scale with the distance from a given $M$-permutohedron.

One can also try to explore further the following simple observation. Similarly to the fact pointed out in Ref.~\cite{sacchi2005entanglement} that entanglement can enhance the distinguishability of entanglement-breaking channels, we see that coherence can enhance distinguishability of completely decohering channels. As a particular example consider the following dual quantum channels,
\begin{equation}
	\Phi^{(1)}(\cdot)=\D(Y(\cdot)Y^\dagger),\quad \Phi^{(2)}(\cdot)=\D(Y^\dagger(\cdot)Y),
\end{equation}
with
\begin{equation}
	Y=\frac{1}{\sqrt{2}} 
	\begin{pmatrix}
		1 & 1  \\
		-1 & 1
	\end{pmatrix}.
\end{equation}
Both these channels have the same classical action and are completely decohering, meaning that the output of both $\Phi^{(1)}$ and $\Phi^{(2)}$ is the same for every incoherent input state. At the same time, we see that a state \mbox{$\ket{+}\propto\ket{0}+\ket{1}$} allows one to perfectly distinguish between $\Phi^{(1)}$ and $\Phi^{(2)}$, as they send it to orthogonal states $\ket{0}$ and $\ket{1}$. This extends the initial idea of Ref.~\cite{sacchi2005entanglement} that distinguishability of resource destroying maps~\cite{liu2017resource} can be improved by using resource states.
		
Last, but not least, from the resource-theoretic perspective one may be interested in quantifying the amount of resources needed to distinguish between classically indistinguishable states and channels. Recall that classical constraints may arise either through a lack of phase reference in the presence of a superselection rule~\cite{bartlett2007reference}, or in the scenarios studied within the resource theory of coherence~\cite{baumgratz2014quantifying,levi2014quantitative}. One can then ask about minimal amounts of resources, e.g., a minimal size of a phase reference, allowing one to overcome those constraints and perform a perfect distinguishability protocol.

\section*{Acknowledgements}

We have a pleasure to thank Francesco Buscemi, Christopher Chubb, David Jennings and Karl Svozil for inspiring discussions. We are also grateful to Wojciech Bruzda for identifying isolated complex Hadamard matrices for dimensions $d \le 17$, to Grzegorz Rajchel for using his code to numerically verify whether a given bistochastic matrix of order four is unistochastic, and to Stanis{\l}aw Pajka for constructing a useful model to demonstrate the coherification of bistochastic maps -- see Fig.~\ref{fig:channel_coher}. We acknowledge financial support from the ARC via the Centre of Excellence in Engineered Quantum Systems, project number CE170100009 (K.K.) and Polish National Science Centre under the project numbers 2016/22/E/ST6/00062 (Z.P.) and 2015/18/A/ST2/00274 (K.{\.Z}.).

\appendix

\section{Visualizing the coherification procedure}
\label{app:coher}

\subsection {Quantum states}

Looking for a coherification of a classical state $\v{p}$ we aim at finding its preimage with respect to the completely decohering channel $\D$, i.e., a quantum state $\rho$ such that $\diag{\rho}=\v{p}$~\cite{korzekwa2018coherifying}. In other words, different coherifications of $\v{p}$ correspond to different quantum states with a fixed diagonal (representing populations in the distinguished basis $\{\ket{i}\}$) but different off-diagonal terms (representing coherences with respect to $\{\ket{i}\}$). The \emph{strength} of coherification can be measured, for example, by the purity or $l_1$-norm of coherence of the coherified state. However, positivity of $\rho$ constrains the off-diagonal terms and the extreme case, which we refer to as complete coherification, corresponds to a pure state $\ket{\psi}$. Now, a set of all complete coherifications of $\v{p}$ is a set of all pure states $\ket{\psi}$ satisfying $\v{p}=\diag{\ketbra{\psi}{\psi}}$. Therefore, the complete coherification procedure can be seen as \emph{quantization} of the simplex of classical probability vectors -- the set of all classical input states is mapped to the set of all pure quantum states while preserving the measurement statistics in the distinguished basis. Similarly, if we constrain the strength of coherification, we will map the classical simplex to a set of mixed quantum states corresponding to partially decohered pure states.

Let us visualize this concept using the simplest example of a qubit system. A distribution over a classical bit can be represented by a unit segment with extremal points corresponding to sharp distributions $(1,0)$ and $(0,1)$, see Fig.~\ref{fig:qubit_coher}. We can now embed this classical state space into a quantum one, i.e., the interval $[0,1]$ representing classical probabilistic states becomes embedded inside the $3$-dimensional Bloch ball containing the density matrices of size $d=2$. This is visualized in Fig.~\ref{fig:qubit_coher} as inserting the unit segment into a balloon. Now, the coherification procedure can be understood as inflating the balloon, effectively expanding the state space. Observe that the classical pure states, $(1,0)$ and $(0,1)$, do not change their positions and become quantum basis states, $|0\rangle$ and $|1\rangle$. From the presented picture it is clear that coherification can be seen as inverse of the decoherence process, which leads to the diminishing of the off-diagonal entries of the density matrix, see Fig.~\ref{fig:qubit_coher}.

\begin{figure}[t]
	\includegraphics[width=\columnwidth]{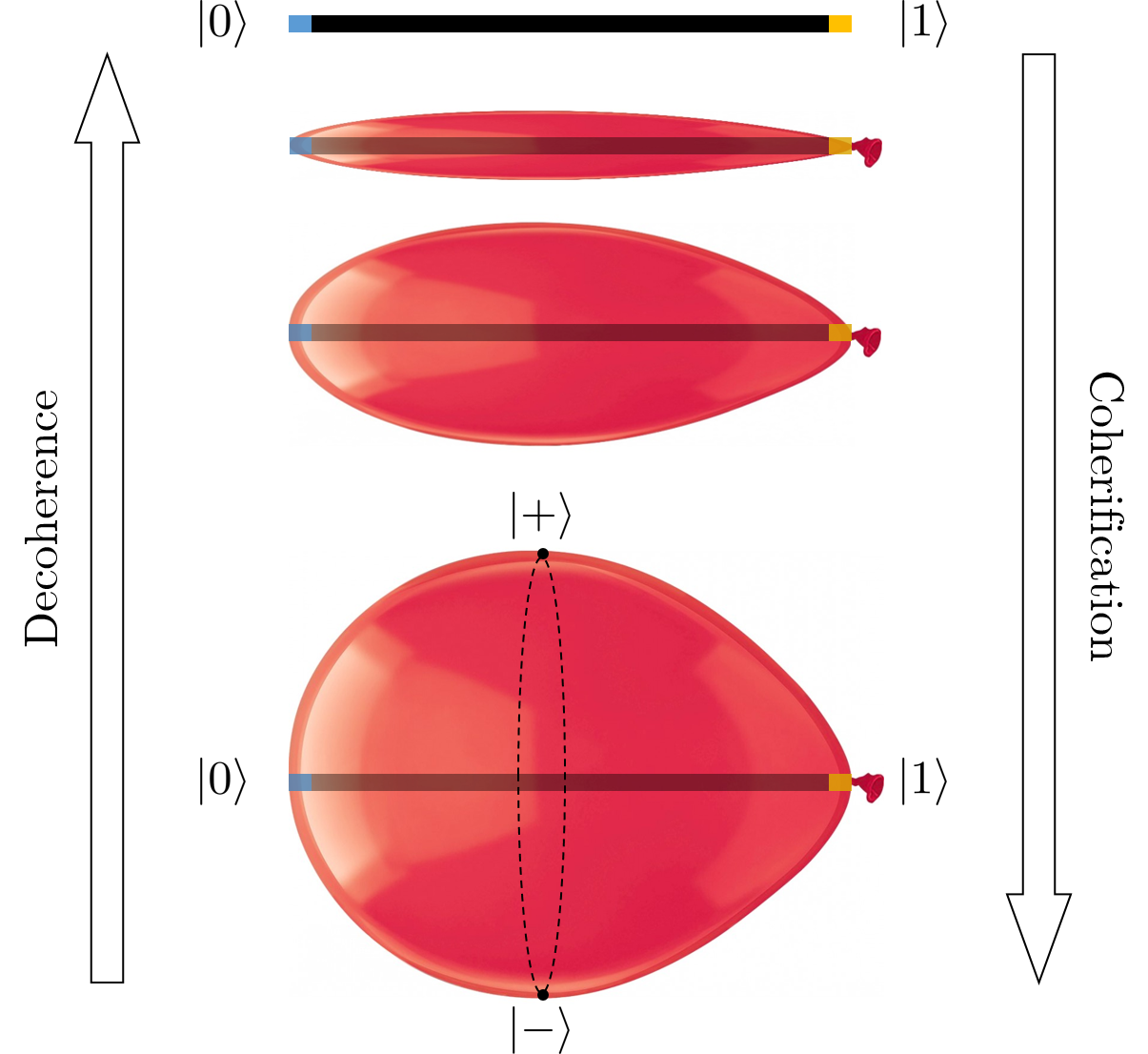}
	\caption{\label{fig:qubit_coher} \emph{Coherification of a qubit.} Probabilistic states of a classical bit (represented by a unit segment with endpoints given by sharp distributions) can be embedded in a quantum state space of density matrices of size $d=2$ (represented by a red balloon). Coherification procedure for qubit systems (visualized by inflating the balloon) continuously expands the state space from a classical simplex to the Bloch sphere of pure states. As such, it can be seen as the inverse of a decohering process which, eventually, brings any quantum state $\rho$ back to the diagonal matrix representing a classical state $\v{p}={\rm diag}(\rho)$.}
	\end{figure}

\begin{figure}[t]
	\begin{tikzpicture}			
		\node at (0\columnwidth,0\columnwidth) {\includegraphics[width=0.6\columnwidth]{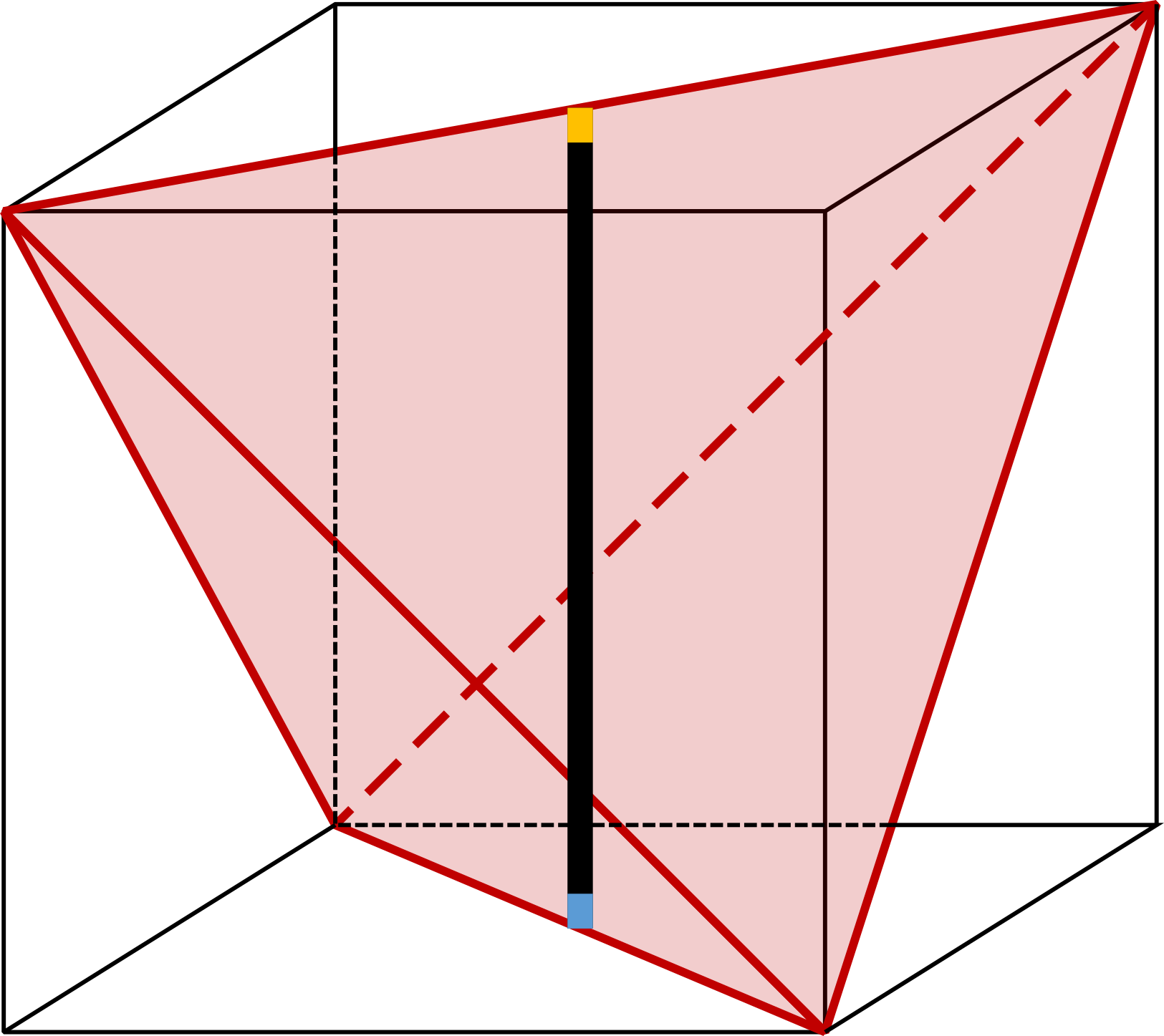}};
		\node at (0.13\columnwidth,-0.31\columnwidth) {\color{black}$\sigma_x(\cdot)\sigma_x$};
		\node at (-0.2\columnwidth,-0.15\columnwidth) {\color{black}$\sigma_y(\cdot)\sigma_y$};
		\node at (0.3\columnwidth,0.3\columnwidth) {\color{black}$\sigma_z(\cdot)\sigma_z$};
		\node at (-0.33\columnwidth,0.17\columnwidth) {\color{black}$\I$};
		\node at (-0\columnwidth,0.24\columnwidth) {\color{black}$\D$};
		\node at (-0.03\columnwidth,-0.245\columnwidth) {\color{black}$\sigma_xD(\cdot)\sigma_x$};
	\end{tikzpicture}
	\caption{\label{fig:qubit_unital} \emph{The set of unital qubit channels.} The set of unital quantum channels acting on a qubit system forms a regular tetrahedron spanned by the identity channel $\I$ and three unitary Pauli channels. It contains a one-dimensional set of classical channels (bistochastic matrices) given by the interval that joins completely decohering map $\D$ and $\D$ followed by a permutation, i.e., by a Pauli $x$ channel. These extremal classical channels correspond to equal mixtures of $\I$ and $\sigma_z(\cdot)\sigma_z$, and equal mixtures of $\sigma_x(\cdot)\sigma_x$ and $\sigma_y(\cdot)\sigma_y$, respectively.}
\end{figure}
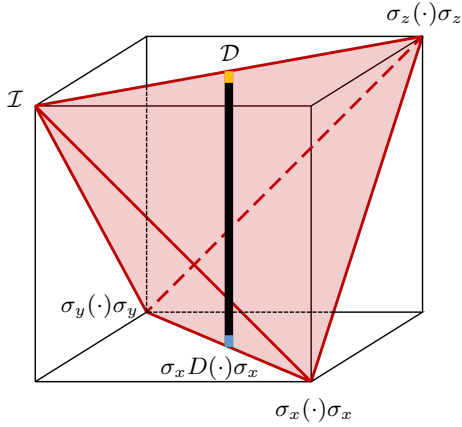

\subsection {Quantum channels}

Looking for a coherification of a classical stochastic matrix $T$ we aim at finding the preimage of its Jamio{\l}kowski state with respect to the completely decohering channel~$\D$, i.e., we look for a quantum channel $\Phi$ such that its Jamio{\l}kowski state $J_\Phi$ satisfies $\diag{J_\Phi}=\vect{T}$~\cite{korzekwa2018coherifying}. In other words, different coherifications of $T$ correspond to different quantum channels with a fixed classical action (representing population transitions in the distinguished basis $\{\ket{i}\}$) but varying otherwise, e.g., with different action on the off-diagonal terms. Note that, although coherification of quantum channels is defined via the coherification of corresponding quantum states, due to an additional trace-preserving constraint the process is more involved, and complete coherification is generally impossible~\cite{korzekwa2018coherifying}. Nevertheless, the coherification procedure can again be seen as \emph{quantization} of the classical space of stochastic matrices -- first, one embeds this space in the space of quantum channels, and then maps every stochastic matrix $T$ into a channel with classical action $T$.

Unlike the set of one-qubit quantum states, which has only three dimensions and can thus be conveniently visualized, the set of all one-qubit quantum channels has 12 dimensions, which makes it hard to analyze. Fortunately, every unital channel acting on a qubit system is unitarily equivalent to a Pauli channel, 
\begin{equation}
	\Psi_{\v{p}}(\cdot)=\sum_{j=0}^3 p_i \sigma_j (\cdot) \sigma_j,
\end{equation}
with $\sigma_j$ denoting three Pauli matrices appended by the identity matrix, $\sigma_0=\iden$, and $\v{p}$ being a classical probability vector of length four. Thus the set all Pauli channels (and, hence, the set of all unital channels) can be represented by a regular $3$-dimensional tetrahedron, which can be easily visualized, see Fig.~\ref{fig:qubit_unital}.

Classical unital channels correspond to bistochastic matrices which, in the case of a 2-dimensional system, can be parametrized by a single number $a\in[0,1]$,
\begin{equation}
	B_a=
	\begin{pmatrix}
		a & 1-a  \\
		1-a & a
	\end{pmatrix}.
\end{equation}
These classical channels, after embedding in the space of unital quantum channels, form an interval within the tetrahedron of unital quantum channels, see Fig.~\ref{fig:qubit_unital}. The endpoints of the interval, $B_1$ and $B_0$, correspond to a completely decohering channel $\D$ and $\D$ followed by the Pauli $x$ channel. To see this, note that the Jamio{\l}kowski state of a classical channel $B_1$ is given by \mbox{$J_{B_1}\propto \ketbra{00}{00}+\ketbra{11}{11}$}, while \mbox{$J_\I=\ketbra{\Omega}{\Omega}$} and \mbox{$J_{\sigma_z}=\ketbra{\Omega'}{\Omega'}$} with \mbox{$\ket{\Omega'}\propto\ket{00}-\ket{11}$}. It is thus clear that \mbox{$J_{B_1}=\frac{1}{2}(J_\I+J_{\sigma_z})$}, so that the classical channel $B_1$ is given by the equal mixture of identity and Pauli $z$ channels, which in turn is equal to~$\D$. Similarly, the Jamio{\l}kowski state of a classical channel $B_0$ is given by \mbox{$J_{B_0}\propto \ketbra{01}{01}+\ketbra{10}{10}$}, while \mbox{$J_{\sigma_x}=\ketbra{\omega}{\omega}$} and \mbox{$J_{\sigma_y}=\ketbra{\omega'}{\omega'}$} with \mbox{$\ket{\omega}\propto\ket{01}+\ket{10}$} and \mbox{$\ket{\omega'}\propto\ket{01}-\ket{10}$}. Analogously, we have that \mbox{$J_{B_0}=\frac{1}{2}(J_{\sigma_x}+J_{\sigma_y})$}, so that the classical channel $B_0$ is given by the equal mixture of Pauli $x$ and $y$ channels, which in turn is equal to $\sigma_x\D(\cdot)\sigma_x$.

To visualize the coherification procedure of the set of classical bistochastic maps we may again imagine inserting the unit interval (representing classical channels) inside a balloon and inflating it. This time, however, the balloon is confined inside the regular tetrahedron of unital channels, see Fig.~\ref{fig:channel_coher}. In practice, it is hardly possible to inflate the balloon so that it reaches the corners of the tetrahedron, which corresponds to complete coherification of classical bistochastic channels to unitary channels. Note also that, similarly to quantum states, coherification of quantum channels can be seen as the inverse process to the decohering supermap (decohering the Jamio{\l}kowski state of a channel), which in the current case sends all elements of the tetrahedron back to the unit interval of classical bistochastic matrices, see Fig.~\ref{fig:channel_coher}.

\begin{figure}[t]
	\begin{center}
		\includegraphics[width=\columnwidth]{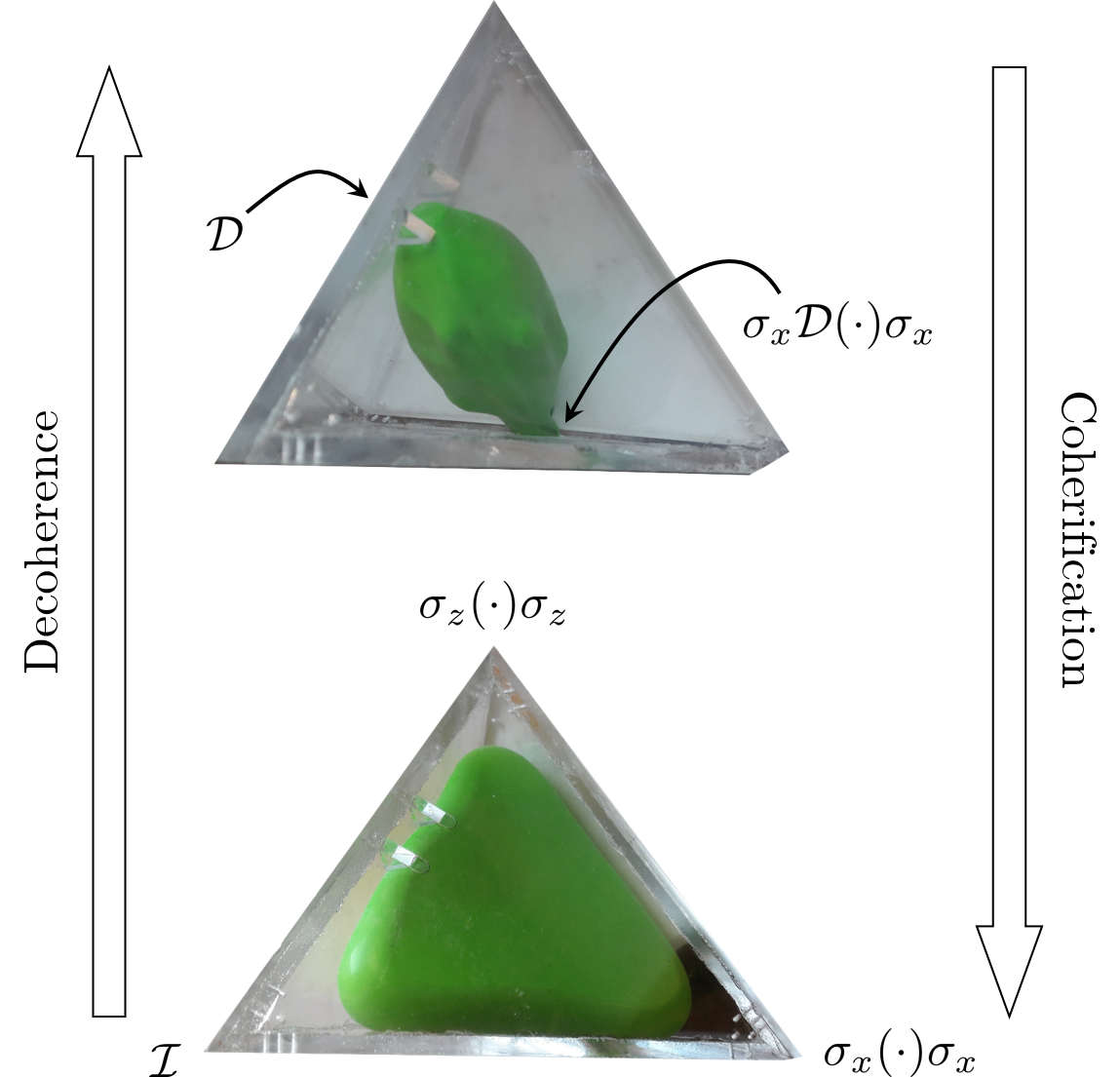}
		\caption{\emph{Coherification of unital qubit channels.} Coherification of the set of classical bistochastic matrices, which expands the unit interval representing them to the full tetrahedron of quantum unital channels, can be visualized by inflating a balloon on a stick inside a tetrahedron. Decoherence in the space of quantum channels shrinks the tetrahedron back to the unit interval.}
		\label{fig:channel_coher}		
	\end{center}
\end{figure}

\section{Proof of Proposition~\ref{prop:insufficient}}
\label{app:insufficient}

\begin{proof}
	Consider a $d$-dimensional probability vector
	\begin{equation}
		\v{p}=\frac{1}{d-1}\left(\frac{2}{3},\frac{2}{3},\frac{2}{3},1,\dots,1\right),
	\end{equation}
	with $d>2$ being even. We will prove that despite the fact that $\v{p}\in\P_d^{d-1}$ (so that it satisfies the necessary condition of Proposition~\ref{prop:necessary}), there does not exist \mbox{$d-1$} perfectly distinguishable states with a fixed classical version $\v{p}$. First, note that $\v{p}$ lies at the boundary of $\P_d^{d-1}$, so that due to Lemma~\ref{lem:boundary} we can restrict our considerations to pure states. This, via Lemma~\ref{lem:dist_uni}, means that finding $d-1$ vectors with classical version $\v{p}$ is equivalent to the unistochasticity of the following $d\times d$ matrix,
	\begin{equation}
		T = \frac{1}{d-1}\renewcommand*{\arraystretch}{1.25}
		\begin{pmatrix}
			\frac{2}{3} &  \cdots & \frac{2}{3} & \frac{d-1}{3} \\
			\frac{2}{3} &  \cdots & \frac{2}{3} & \frac{d-1}{3} \\
			\frac{2}{3} &  \cdots & \frac{2}{3} & \frac{d-1}{3} \\
			1 &  \cdots & 1 & 0 \\
			\vdots    & \ddots & \vdots & \vdots \\
			1 &  \cdots & 1 & 0 \\
		\end{pmatrix}.
	\end{equation}
	We will now assume that $T$ is unistochastic, with the corresponding unitary matrix denoted by $U$, and show that for even $d$ this leads to a contradiction. 
	
	Without loss of generality we can assume that the first row and the first column of $U$ are real and nonnegative. Due to orthogonality of the first and last column of $U$, the first three elements of the last column of $U$ are equal to either $\v{c}$ or $\v{c}'$ with
	\begin{subequations}
		\begin{align}
		\v{c} &= \sqrt{\frac{d-1}{3}}\left(1,e^{i \frac{2\pi}{3}},e^{-i \frac{2\pi}{3}} \right),\\
		\v{c}' &= \sqrt{\frac{d-1}{3}}\left(1,e^{-i \frac{2\pi}{3}},e^{i \frac{2\pi}{3}} \right).
		\end{align}
	\end{subequations}
	We can assume that the last column of $U$ is specified by~$\v{c}$, as the proof for the alternative choice is analogous. Similarly, for $1<i<d$, due to orthogonality of the $i$-th and last column, the first three elements of the $i$-th column are equal to either $\v{a}$ or $\v{b}$ with
	\begin{subequations}
		\begin{align}
		{\v{a}} &= \sqrt{\frac{2}{3}}\left( 1,1,1 \right) \\
		{\v{b}} &= \sqrt{\frac{2}{3}}\left(1,e^{-i \frac{2\pi}{3}},e^{i \frac{2\pi}{3}} \right).
		\end{align}
	\end{subequations}
	
	The postulated matrix $U$ (up to permutation of columns) has consequently the following form
	\begin{align}
	\!\!\!U = \frac{1}{\sqrt{d-1}} 
	\left(
	\begin{array}{ccc|ccc|c}
	\v{a}  &   \ldots &  \v{a}  & 
	\v{b}  &   \ldots & \v{b}  &  \v{c} \\
	\hline 
	\multicolumn{3}{c|}{ } & \multicolumn{3}{c|}{}  & 0 \\
	\multicolumn{3}{c|}{A} & \multicolumn{3}{c|}{B} & \vdots \\
	\multicolumn{3}{c|}{ } & \multicolumn{3}{c|}{}  & 0
	\end{array}
	\right)\!,\!\!
	\end{align}
	with $A$ being a $(d-3) \times d_A$ matrix, $B$ being a $(d-3) \times d_B$ matrix, and \mbox{$d_A+d_B=d-1$}. Orthogonality of the first $d_A$ columns of $U$ implies 
	\begin{equation}
	\braket{a_k}{a_l} = 
	\begin{cases}
	-2 & \text{for} \quad k \neq l\\
	d-3 & \text{for} \quad k=l 
	\end{cases},
	\end{equation}
	with the vectors $\ket{a_k}$ denoting the columns of the $A$. Analogous condition holds for the columns of $B$, denoted by $\ket{b_k}$.
	
	It is now straightforward to compute $A^{\dagger} A $ and $B^{\dagger} B $, yielding
	\begin{subequations}
		\begin{align}
		A^{\dagger} A & = (d-1) \1_{d_A} - 2d_A \ket{+_{d_A}}\bra{+_{d_A}} ,\\
		B^{\dagger} B & = (d-1) \1_{d_B} - 2d_B \ket{+_{d_B}}\bra{+_{d_B}} ,
		\end{align}
	\end{subequations}
	where
	\begin{equation}
	\ket{+_{d_A}}= \frac{1}{\sqrt{d_A}} \sum_k\ket{k},
	\end{equation}
	and equivalently for $d_B$.	We conclude that
	\begin{subequations}
		\begin{align}
		\label{eq:rankA}
		\rank A &= 
		\begin{cases}
		d_A &\text{for} \quad 2d_A+1 \neq d,\\
		d_A-1 & \text{for} \quad  2d_A+1 = d,
		\end{cases}\\
		\label{eq:rankB}	
		\rank B &= 
		\begin{cases}
		d_B &\text{for} \quad 2d_B +1 \neq d,\\
		d_B-1 & \text{for} \quad  2d_B+1 = d.
		\end{cases}
		\end{align}
	\end{subequations}
	
	The final orthogonality relations are between column $k$ from the first block, $k\leq d_A$, and column $l$ from the second block, $l>d_A$. They give
	\begin{align}
		0 &= \v{a}\cdot\v{b} + \braket{a_k}{b_{l-d_A}} = \braket{a_k}{b_{l-d_A}},
	\end{align}
	which means that $A^{\dagger}B=0$. In other words 
	\begin{equation} 
		\text{span} \{ \ket{a_k} \}_{k=1}^{d_A}  \perp \text{span} \{ \ket{b_l} \}_{l=1}^{d_B},
	\end{equation}
	so that
	\begin{equation}
		\rank{A}+\rank{B}=\rank(A|B),
	\end{equation}
	where $A|B$ is a $(d-3)\times(d-1)$ matrix build by concatenating matrices $A$ and $B$. Since $\rank(A|B)\leq 3$, we have
	\begin{equation}
		\rank{A}+\rank{B} \leq d-3,
	\end{equation}
	which, due to Eqs.~\eqref{eq:rankA}-\eqref{eq:rankB}, is impossible for an even dimension $d$. Thus, the postulated matrix $U$ cannot exist.
\end{proof}

\section{Distinguishability region $\A_4^3$}
\label{app:tetra_skeleton}

In this appendix we describe the structure of distinguishability region $\A_4^3$. We first show that the subset of permutohedron~$\P_4^3$ presented in Fig.~\hyperref[fig:tetra_skeleton]{2a} does belong to~$\A_4^3$. We then prove that the subset of permutohedron~$\P_4^3$ presented in Fig.~\hyperref[fig:tetra_skeleton]{2b} does not belong to~$\A_4^3$. Finally, we formulate a conjecture on the exact form of~$\A_4^3$ (presented in Fig.~\hyperref[fig:tetra_skeleton]{2c}) and support it numerically.

\subsection{Subset of $\P_4^3$ belonging to $\A_4^3$}

The permutohedron $\P_4^3$ is a tetrahedron with vertices given by $\v{f}^1=\frac{1}{3}(1,1,1,0)$ and $\v{f}^i$ for $i\in\{2,3,4\}$ are given by permutations of $\v{f}^1$. Without loss of generality, a point on the edge of this tetrahedron has the form \mbox{$\v{p}^{(s)} = \frac{1}{3} \left( s,1-s,1,1\right)$} with $s \in \left[0,\frac{1}{2}\right]$. We will consider a point $\v{p}^{(s,t)}$ on a line connecting $\v{p}^{(s)}$ and the centre $\v{\eta}$ of $\P_4^3$, 
\begin{equation} 
	\v{p}^{(s,t)} = 3\left(\frac{1}{3} - t \right)\v{\eta} + 3 t \v{p}^{(s)},
\end{equation} 
with $t \in \left[0,\frac{1}{3}\right]$. We will now show that all such points belong to the distinguishability region $\A_4^3$. In order to achieve this we will consider the following three pure states
\begin{align}
	\ket{\psi_1} & = x_1\ket{1}+x_2\ket{2}+ x_3 \ket{3}+x_4\ket{4}, \\
	\ket{\psi_2} & = x_1\ket{1}+x_2e^{i \alpha_2}\ket{2}+ x_3e^{i \alpha_3} \ket{3}+x_4e^{i \alpha_4}\ket{4}, \\
	\ket{\psi_3} & = x_1\ket{1}+x_2e^{i \alpha_2}\ket{2}+ x_3e^{i \alpha_4} \ket{3}+x_4e^{i \alpha_3}\ket{4},
\end{align}
with $x_i = \sqrt{p^{(s,t)}_{i}}$ and prove that for all $s\in \left[0,\frac{1}{2}\right]$ and $t \in \left[0,\frac{1}{3}\right]$ there exists a choice of phases \mbox{$\{\alpha_2,\alpha_3,\alpha_4\}$}, such that the above states are mutually orthogonal. 

The overlap $\braket{\psi_2}{\psi_3}$ reads
\begin{equation}
	\braket{\psi_2}{\psi_3} = \frac{1}{2}\left( 1 - t + \left( 1+t\right) \cos(\alpha_3 - \alpha_4) \right),
\end{equation}  
so that orthogonality condition, $\braket{\psi_2}{\psi_3}=0$, gives
\begin{equation}
	\label{alpha3_condition}
	\alpha_3 = \arccos\frac{t-1}{t+1} + \alpha_4.
\end{equation}
The remaining overlaps are equal, \mbox{$\braket{\psi_1}{\psi_2}\!=\! \braket{\psi_1}{\psi_3}\!=:\! F$}, and given by
\begin{align}
	F=&\frac{1}{4}[1 + \left( 4s - 3 \right)t
	+(1+t-4st)e^{i \alpha_2}  \nonumber\\ 
	&+\left. \left(1 + t\right)(e^{i \alpha_3} +e^{i \alpha_4})\right].
	\end{align}
Using Eq.~\eqref{alpha3_condition} we can simplify the above expression to arrive at
\begin{align}
	\label{eq:annulus}
	F = &\frac{1}{4}[1 + \left( 4s - 3 \right)t +  (1 + t -4st)e^{i \alpha_2} \nonumber\\ 
	& + 2\left(i\sqrt{t} +t \right)e^{i \alpha_4}].
\end{align}

We now note that Eq.~\eqref{eq:annulus} for all $\alpha_2,\alpha_4 \in [0,2\pi)$ describes an annulus $\mathrm{ann}(x;R,r)$ with the centre $x$, larger radius $R$ and smaller radius $r$ equal to
\begin{eqnarray}
	x &=&\frac{1}{4} \left( 1 + \left( 4s - 3 \right)t \right),\label{eq:ann1}\\
	R &=&\frac{1}{4}\left| \left| 1 + t -4st \right|+ 2 \left| i\sqrt{t} +t  \right| \right|\label{R},\label{eq:ann2}\\
	r &=& \frac{1}{4}\left|\left|1 + t -4st\right| - 2\left|i\sqrt{t} +t \right| \right| \label{r}.\label{eq:ann3}
\end{eqnarray}
	The existence of phases such that $\{\ket{\psi_1},\ket{\psi_2},\ket{\psi_3}\}$ are mutually orthogonal is thus equivalent to \mbox{$0 \in \mathrm{ann}(x;R,r)$}, which can be verified by the following elementary calculations.
	
	First, we need to prove that the distance between the centre of the annulus $x$ and 0 is smaller than the larger radius $R$ (see Fig.~\ref{fig:annulus}). Substituting the expressions for $x$ and $R$ into $x \leq R$ yields
	\begin{align}
		4t(1-2s) + 2\sqrt{t(t+1)}\geq 0,
	\end{align} 
	which is always satisfied since $s\leq 1/2$ and $t\geq 0$.
	Next, we need to prove that the distance between the centre of the annulus $x$ and 0 is larger than the smaller radius $r$. Here we have two cases: one when \mbox{$\left|1 + t -4st\right| \geq 2\left|i\sqrt{t} +t \right|$} holds, and one when the opposite holds. In both cases the condition $x\geq r$ simplifies to $t \leq 1/3$ which is always true, because $t\in[0,1/3]$.
	\begin{figure}
	\begin{center}
		\includegraphics[width=0.65\columnwidth]{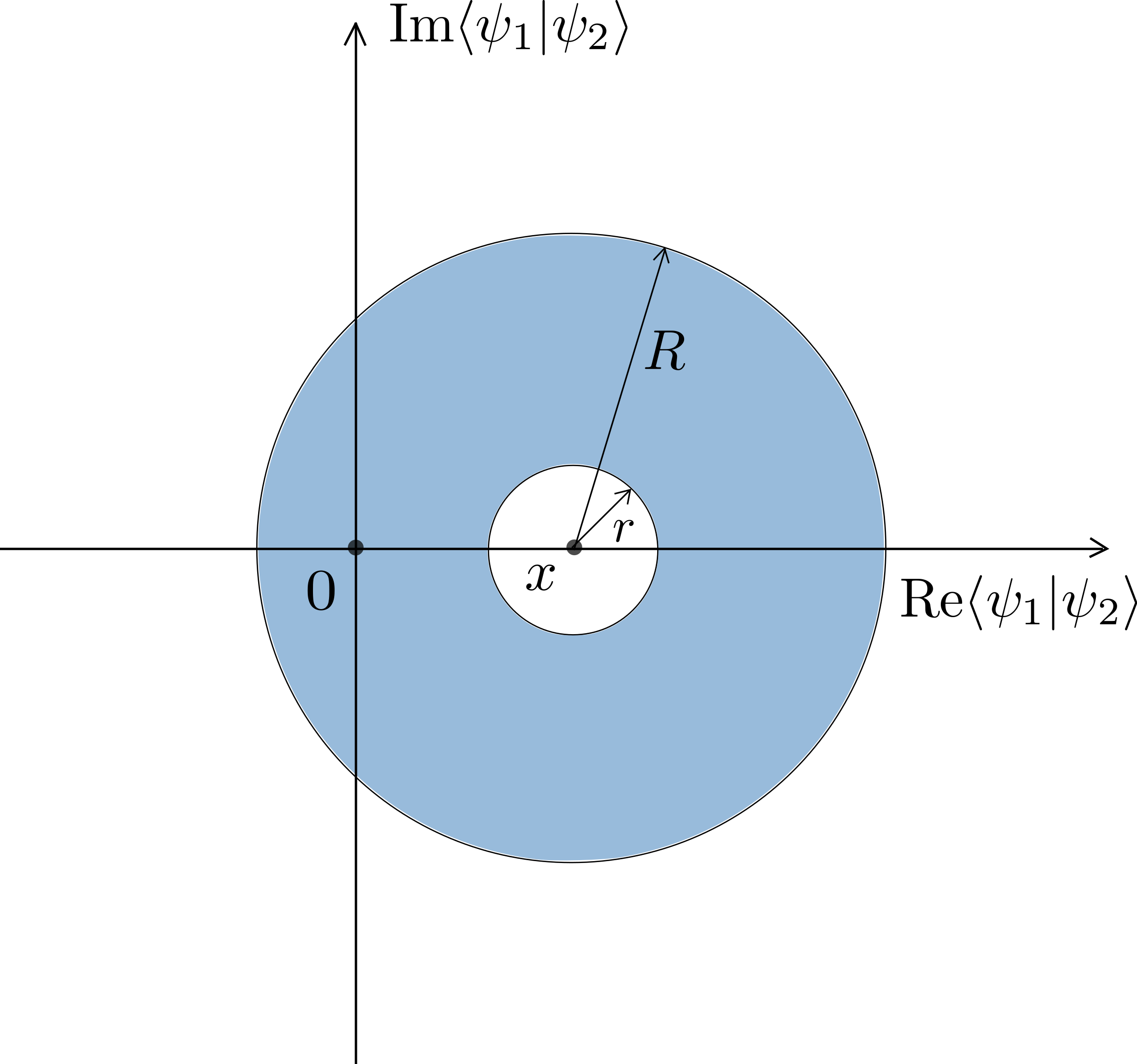}
		\caption{\label{fig:annulus}Annulus  $\mathrm{ann}(x;R,r)$ described by Eqs.~\eqref{eq:ann1}-\eqref{eq:ann3}}.
	\end{center}
\end{figure}

\subsection{Subset of $\P_4^3$ not belonging to $\A_4^3$}

Without loss of generality, a point $\v{p}$ lying in the interior of the face of permutohedron $\P_4^3$ can be expressed by
\begin{align}
	\v{p} &= \frac{1}{3}(1,r+s,q+s,q+r),
\end{align}
with $q+r+s=1$ and $q,r,s>0$. Changing variables according to
\begin{align}
	a_1 =\frac{r+s}{2},\quad a_2 =\frac{q+s}{2},\quad a_3 =\frac{q+r}{2},
\end{align}
we have
\begin{equation}
	\v{p}=\frac{1}{3}(1,2a_1,2a_2,2a_3),
\end{equation}
with $a_1+a_2+a_3=1$, $0<a_j<\frac{1}{2}$. 

We will now prove that one cannot find three mutually orthogonal states \mbox{$\{\ket{\psi_1},\ket{\psi_2},\ket{\psi_3}\}$} with the classical version given by $\v{p}$. The general form of such states is given by
\begin{align}
	\ket{\psi_1} =& \frac{1}{\sqrt{3}} \left[\ket{0}+\sqrt{2}\sum_{j=1}^3 \sqrt{a_j}\ket{j} \right] \\
	\ket{\psi_2} =& \frac{1}{\sqrt{3}} \left[\ket{0}+\sqrt{2}\sum_{j=1}^3 \sqrt{a_j}e^{i\phi_j}\ket{j} \right] \\
	\ket{\psi_3} =& \frac{1}{\sqrt{3}} \left[\ket{0}+\sqrt{2}\sum_{j=1}^3 \sqrt{a_j}e^{i\theta_j}\ket{j} \right].
	\end{align} 
	Orthogonality conditions can be now rewritten in terms of truncated vectors,
	\begin{align}
	\ket{\tilde{\psi_1}} &=  \sum_{j=1}^{3}\sqrt{a_j} \ket{j},\\ 
	\ket{\tilde{\psi_2}} &=  \sum_{j=1}^{3}\sqrt{a_j} e^{i \phi_j}\ket{j},\label{psi2} \\
	\ket{\tilde{\psi_3}} &=  \sum_{j=1}^{3}\sqrt{a_j} e^{i \theta_j}\ket{j} ,
	\end{align}
	as:
	\begin{equation}
	\label{eq:ortho}
		\langle\tilde{\psi_1}|\tilde{\psi_2}\rangle =	\langle\tilde{\psi_1}|\tilde{\psi_3}\rangle =
		\langle\tilde{\psi_2}|\tilde{\psi_3}\rangle = -\frac{1}{2}.
	\end{equation}
	
	We will now prove, by contradiction, that for the orthogonality condition to hold the set of vectors \mbox{$\{\ket{\tilde{\psi_1}},\ket{\tilde{\psi_2}},\ket{\tilde{\psi_3}}\}$} must be linearly independent. Assume that these vectors are linearly dependent, i.e., there exists complex numbers $\alpha,\beta,\gamma$ such that \mbox{$|\alpha| + |\beta|+ |\gamma| > 0$} and
	\begin{equation}
			\alpha \ket{\tilde{\psi_1}} + \beta \ket{\tilde{\psi_2}} + \gamma \ket{\tilde{\psi_3}} = 0.
	\end{equation}
	By the orthogonality condition, Eq.~\eqref{eq:ortho}, we have
	\begin{align}
		\alpha - \frac{\beta}{2} - \frac{\gamma}{2} &= 0, \\ 
		-\frac{\alpha}{2} + \beta -\frac{\gamma}{2} &= 0, \\ 
		-\frac{\alpha}{2} - \frac{\beta}{2} + \gamma &=0,
	\end{align}
	which leads to conclusion that \mbox{$\alpha = \beta = \gamma$}. Since at least one of them is nonzero, all are nonzero. Thus, we have
	\begin{equation}
		\ket{\tilde{\psi_3}} = - \ket{\tilde{\psi_1}} - \ket{\tilde{\psi_2}}.
	\end{equation}
	But this means:
	\begin{equation}
		\sum_{j=1}^{3} \sqrt{a_j} e^{i \theta_j} \ket{j} = \ket{\tilde{\psi_3}} =
		-   \sum_{j=1}^{3} \sqrt{a_j} \left(1 + e^{i \phi_j} \right)\ket{j}.
	\end{equation}
	Hence,
	\begin{equation}
		\forall j \in \{1,2,3\}: \quad e^{i \theta_j}= - \left(1 + e^{i \phi_j} \right).
	\end{equation}
	In particular, focusing on the absolute value,
	\begin{equation}
			\label{cos}
			\forall j \in \{1,2,3\}: \quad \cos \phi_j = -\frac{1}{2},
	\end{equation}
	which implies
	\begin{equation}
		\label{sin}
		\forall j \in \{1,2,3\}: \quad \sin \phi_j = \epsilon_j  \frac{\sqrt{3}}{2},
	\end{equation}
	where $\epsilon_j \in \{-1,1\}$. Inserting Eqs.~(\ref{cos})-(\ref{sin}) into Eq.~(\ref{psi2}), we obtain
	\begin{equation}
		\ket{\tilde{\psi_2}} =   \sum_{j=1}^{3} \sqrt{a_j} \left(-\frac{1}{2} + \epsilon_j i \frac{\sqrt{3}}{2}  \right) \ket{j}.
	\end{equation}
	The overlap $\langle\tilde{\psi_1}|\tilde{\psi_2}\rangle$ is therefore equal to 
	\begin{align}
		\langle\tilde{\psi_1}|\tilde{\psi_2}\rangle & = \sum_{j=1}^{3} {a_j} \left(-\frac{1}{2} + \epsilon_j i \frac{\sqrt{3}}{2}  \right) \nonumber \\
		&= - \frac{1}{2} + i \frac{\sqrt{3}}{2} \sum_{j=1}^{3} a_j \epsilon_j.
	\end{align}
	Since we know that $\langle\tilde{\psi_1}|\tilde{\psi_2}\rangle=-\frac{1}{2}$ the imaginary
	part must vanish, meaning that
	\begin{equation}
		\sum_{j=1}^{3} a_j \epsilon_j = 0.
		\end{equation}
	But this is only possible if $a_i$ are a permutation of $\left(\frac{1}{2},\frac{1}{4},\frac{1}{4}\right)$. As we know that $a_j < \frac{1}{2}$ we arrive at a contradiction, which disproves  the assumption of linear dependence of \mbox{$\{\ket{\tilde{\psi_1}},\ket{\tilde{\psi_2}},\ket{\tilde{\psi_3}}\}$}.

	We now know that the set \mbox{$\{\ket{\tilde{\psi_1}},\ket{\tilde{\psi_2}},\ket{\tilde{\psi_3}}\}$} must be linearly independent but, on the other hand, the Gram matrix for these vectors reads 
	\begin{equation}
	\begin{pmatrix}
	1 & -\frac{1}{2} & -\frac{1}{2} \\
	-\frac{1}{2}& 1 & -\frac{1}{2} \\
	-\frac{1}{2}& -\frac{1}{2}&1
	\end{pmatrix}.
	\end{equation}
	It is straightforward to check that it is of rank 2, which means that vectors \mbox{$\{\ket{\tilde{\psi_1}},\ket{\tilde{\psi_2}},\ket{\tilde{\psi_3}}\}$} are linearly dependent. This finishes the proof that orthogonal \mbox{$\{\ket{\tilde{\psi_1}},\ket{\tilde{\psi_2}},\ket{\tilde{\psi_3}}\}$} cannot exist.
	
	We now proceed to states with classical version $\v{q}$ lying on the lines connecting the centre of the permutohedron $\P_4^3$ with the centres of its faces. Without loss of generality such a point can be expressed by 
	\begin{align}
		\v{q} &= \frac{1}{4}(1-3t,1+t,1+t,1+t),
	\end{align}
	with $t\in[-1/9,0]$. Now, via Lemma~\ref{lem:dist_uni}, we know that the existence of mutually orthogonal \mbox{$\{\ket{{\psi_1}},\ket{{\psi_2}},\ket{{\psi_3}}\}$} with classical version $\v{q}$ is equivalent to the existence of a unistochastic matrix
	\begin{equation}
		T(\v{q})=
		\frac{1}{4}\begin{pmatrix}
		1-3t & 1-3t& 1-3t & 1+9t\\
		1+t & 1+t  & 1+t &1-3t \\
		1+t & 1+t  & 1+t &1-3t\\
		1+t & 1+t  & 1+t & 1-3t 
		\end{pmatrix}.
	\end{equation}
	The above family of bistochastic matrices has been studied in Ref.~\cite{bengtsson2005birkhoff}, where the authors showed that for no values of $t\in[-1/9,0]$ is $T(\v{q})$ unistochastic (see Eq.~(33) of Ref.~\cite{bengtsson2005birkhoff}). Therefore, there cannot exist three orthogonal pure states with classical action given by $\v{q}$.

\subsection{Conjectured form of $\A_4^3$}

Due to Lemma~\ref{lem:dist_uni}, the problem of finding $d-1$ orthogonal vectors with the same classical version $\v{p}$ is equivalent to verifying whether a particular bistochastic matrix $B$, with the first $d-1$ columns given by $\v{p}$, is unistochastic. By employing the algorithm proposed by Uffe Haagerup and described in Ref.~\cite{rajchel2018robust} we numerically verified whether matrices $B$ corresponding to probability vectors within the permutohedron $\P_4^3$ are unistochastic, and this way obtained a numerical approximation of the distinguishability region $\A_4^3$. Its form suggested the conjecture described by Eq.~\eqref{eq:3dist_conjecture} and the supporting evidence is presented in Fig.~\ref{fig:conjecture_numerics}.

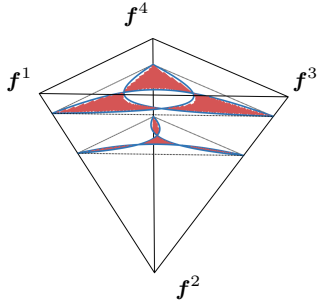
\begin{figure}[t]
	\begin{tikzpicture}			
	\node at (0\columnwidth,0\columnwidth) {\includegraphics[width=0.4\columnwidth]{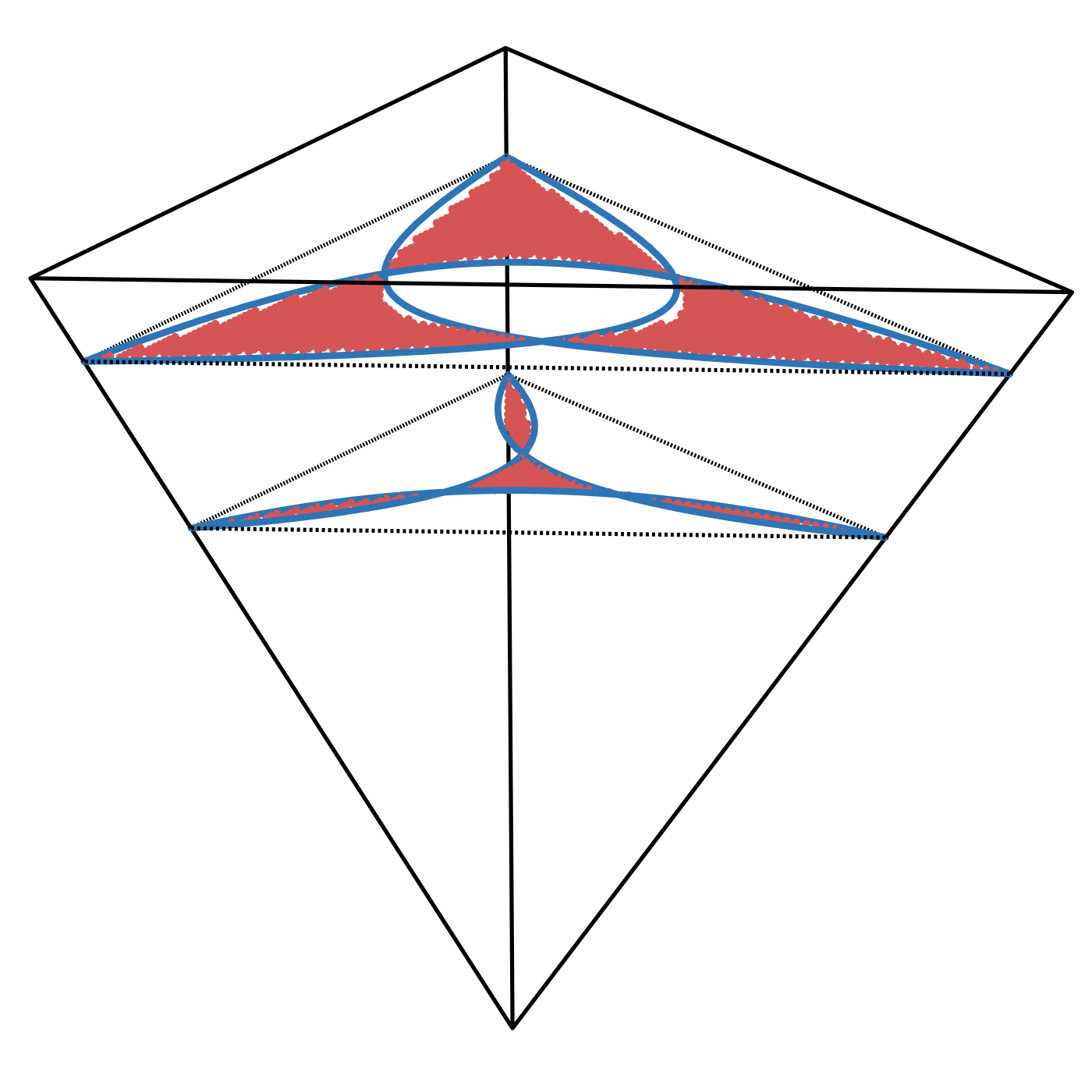}};
	\node at (-0.22\columnwidth,0.12\columnwidth) {\footnotesize\color{black}$\v{f}^1$};
	\node at (0.04\columnwidth,-0.2\columnwidth) {\footnotesize\color{black}$\v{f}^2$};
	\node at (0.22\columnwidth,0.12\columnwidth) {\footnotesize\color{black}$\v{f}^3$};
	\node at (-0.04\columnwidth,0.22\columnwidth) {\footnotesize\color{black}$\v{f}^4$};
	\end{tikzpicture}
	\caption{\label{fig:conjecture_numerics} \emph{Numerical evidence supporting conjectured form of~$\A_4^3$.} Cross sections through the permutohedron $\P_4^3$ with 3-distinguishable states confirmed numerically indicated by red dots, and the conjectured form of $\A_4^3$, described by Eq.~\eqref{eq:3dist_conjecture}, indicated by blue solid curves.}
\end{figure}

\section{$d+1$ perfectly distinguishable unitaries}
\label{app:d+1_unitaries}

Consider $d=p-1$ for prime $p$ and fix a bistochastic matrix $T$,
\begin{equation}
	T=\frac{1}{p} (W+\iden).
\end{equation}
For $d$ for which $T$ is unistochastic we define $U$ to be any unitary satisfying $U\circ U=T$. Next, we define the following $p$ diagonal unitaries $E^{(k)}$:
\begin{equation}
	E^{(k)}_{jj}=\exp\left(\frac{2\pi i (j-1)k}{p}\right).
\end{equation}
Now, $p$ unitaries defined by $E^{(k)}UE^{(k)}$ are all perfectly distinguishable, since Eq.~\eqref{eq:ort_unitaries} is satisfied for all $k,l$ (with $E$ in place of both $L$ and $R$). More precisely, the condition given by Eq.~\eqref{eq:ort_unitaries} reads
\begin{align*}
	0=&\left(\sum_{m=1}^d \exp\left(\frac{2\pi i (m-1)(k-l)}{p}\right)\right)^2\\
	&+\sum_{m=1}^d \exp\left(\frac{4\pi i (m-1)(k-l)}{p}\right).
\end{align*}
Since we sum over all but one $p$-th roots of the unity, the right hand side of the above simplifies to
\begin{equation*}
\exp\left(\frac{4\pi i (p-1)(k-l)}{p}\right)\!-\exp\left(\frac{4\pi i (p-1)(k-l)}{p}\right),
\end{equation*}
and thus vanishes, as required.

\section{Proof of Proposition~\ref{prop:transposition}}
\label{app:transposition}

\begin{proof}
	Without loss of generality we can assume that \mbox{$k=1$}, $l=2$ and $j=2$. Then, columns $k$ and $l$ of $T'$, which we will denote by $\v{x}$ and $\v{y}$, are given by
	\begin{subequations}
		\begin{align}
		\v{x} =& \left( T_{11}, \alpha T_{22} , T_{31},\dots ,T_{d1} \right)^\top, \\
		\v{y} =& \left( \alpha T_{12},  T_{21}, \alpha T_{32},\ldots, \alpha T_{d2} \right)^\top.
		\end{align}
	\end{subequations}	
	By assumption, the entries of the above vectors satisfy the triangle inequality, so that using Proposition~\ref{prop:sufficient} we know that there exist two pairs vectors $\ket{\xi}$, $\ket{\xi^{\prime}}$ and $\ket{\eta}$,$\ket{\eta^{\prime}}$, such that  $\braket{\xi}{\xi^{\prime}}=0$, $\braket{\eta}{\eta^{\prime}}=0$ and
	\begin{subequations}	
		\begin{align}
		\lvert \xi_k \rvert^2 & = \lvert \xi^{\prime}_k\rvert^2 = x_k \\
		\lvert \eta_k \rvert^2 &= \lvert \eta^{\prime}_k \rvert^2 = y_k,
		\end{align}
	\end{subequations}	
	where we used standard shorthand notation \mbox{$\xi_k=\braket{k}{\xi}$}. 
	
	We now define bipartite states
	\begin{subequations}
		\begin{align}
		\ket{\hat{\xi}}  =& \xi_1 \ket{11} + \xi_2 \ket{22} + \xi_3 \ket{31} + \ldots +\xi_{d}   \ket{d1},  \\
		\ket{\hat{\eta}} =& \eta_1 \ket{12} + \eta_2 \ket{21} + \eta_3 \ket{32} + \ldots + \eta_{d} \ket{d2},
		\end{align}
	\end{subequations}
	with analogous definitions for $\ket{\hat{\xi}^{\prime}}$ and $\ket{\hat{\eta}^{\prime}}$, so that all four states are mutually orthogonal. Moreover, we introduce
	\begin{equation}
	S(t) =  \ketbra{0}{0} + t\ketbra{1}{1},
	\end{equation}
	and define another set of bipartite states
	\begin{subequations}
		\begin{align}
		\ket{f}  =& \1 \otimes S \left( \frac{1}{\sqrt{\alpha}} \right) \ket{\hat{\xi}}, \\ 
		\ket{g}  =&\1 \otimes S \left( \frac{1}{\sqrt{\alpha}} \right) \ket{\hat{\eta}},
		\end{align}	
	\end{subequations}
	with analogous definitions for $\ket{f^{\prime}}$ and $\ket{g^{\prime}}$. Note that vectors $\ket{f}$ and $\ket{g}$ are not normalized, but they satisfy
	\begin{equation}
	\braket{f}{f} + \braket{g}{g} = 2.
	\end{equation}
	
	We are now ready to construct channels $\Phi^{(1)}$ and $\Phi^{(2)}$ with classical action $T$ by providing their their Jamio{\l}kowski states,	
	\begin{subequations}
		\begin{align}
		\!\!\!J_{\Phi^{(1)}} \!& = \frac{1}{d} \left(\ketbra{f}{f} + \ketbra{g}{g} + \sum_{k=1}^{d} \sum_{l=3}^{d}T_{kl} \ketbra{kl}{kl} \right)\!, \\
		\!\!\!J_{\Phi^{(2)}} \!&= \frac{1}{d} \left( \ketbra{f^{\prime}}{f^{\prime}} + \ketbra{g^{\prime}}{g^{\prime}} + \sum_{k=1}^{d} \sum_{l=3}^{d}T_{kl} \ketbra{kl}{kl}\right)\!\!.
		\end{align}
	\end{subequations}
	One can verify that $T$ is indeed classical action of $\Phi^{(1)}$ and $\Phi^{(2)}$ by a direct calculation of the diagonal elements of $J_{\Phi^{(1)}}$ and $J_{\Phi^{(2)}}$.
	
	Moreover, the action of $\Phi^{(1)}$ and $\Phi^{(2)}$ on one half of an unnormalized bipartite state $\ket{\psi}$,
	\begin{align}
	\ket{\psi}  = \ket{00} +\sqrt{\alpha} \ket{11} = \1 \otimes S(\sqrt{\alpha}) \ket{\Omega},
	\end{align}
	maps it to two orthogonal states. To see this, let us calculate it explicitly,	
	\begin{align}
	&\!\!\!\!\!\!\!\!\!\left(\Phi^{(1)} \otimes  \I \right) \left(\ketbra{\psi}{\psi} \right) \nonumber\\ 
	&= \left(\Phi^{(1)} \otimes  \I \right) 
	\left[ (\1 \otimes S ( \sqrt{\alpha} ))
	\ketbra{\Omega}{\Omega}  
	( \1 \otimes S ( \sqrt{\alpha} ))\right] \nonumber \\
	&= \left( \1 \otimes S ( \sqrt{\alpha} )\right)
	\left(\Phi^{(1)} \otimes  \1 \right) 
	\left( \ketbra{\Omega}{\Omega} \right)
	\left( \1 \otimes S ( \sqrt{\alpha} )\right)\nonumber  \\
	&= \left( \1 \otimes S ( \sqrt{\alpha} \right) d J_{\Phi^{(1)}}
	\left(  \1 \otimes S ( \sqrt{\alpha} ) \right) \nonumber\\
	&= \left(\1 \otimes S ( \sqrt{\alpha} )\right)
	\left(\ketbra{f}{f} + \ketbra{g}{g} \right)
	\left( \1 \otimes S ( \sqrt{\alpha} )\right)\nonumber \\
	&=  \ket{\hat{\xi}}\!\bra{\hat{\xi}} + \ketbra{\hat{\eta}}{\hat{\eta}},
	\end{align}	 
	and, through analogous calculation, we also get
	\begin{align}
	\left(\Phi^{(2)} \otimes  \I \right) \left( \ketbra{\psi}{\psi} \right)=\ket{\hat{\xi^{\prime}}}\!\bra{\hat{\xi}^{\prime}} + \ket{\hat{\eta}^{\prime}}\!\bra{\hat{\eta}^{\prime}}.
	\end{align}
	Since $\ket{\hat{\xi}}$, $\ket{\hat{\xi}'}$, $\ket{\hat{\eta}}$ and $\ket{\hat{\eta}'}$ are all mutually orthogonal, we conclude that channels $\Phi^{(1)}$ and $\Phi^{(2)}$ are perfectly distinguishable.
\end{proof}

\bibliography{Bibliography_channels}

\end{document}